\documentclass[acmsmall,screen,natbib=false]{acmart} 
\usepackage{graphicx}
\usepackage{xcolor}
\usepackage{threeparttable}
\usepackage{multirow}
\usepackage{booktabs}
\usepackage{amsmath}
\usepackage{amsthm}
\usepackage{stmaryrd}
\usepackage{enumitem}
\usepackage{mathpartir}
\usepackage{tlatex}
\usepackage{mathrsfs}
\usepackage{multicol}
\usepackage{subcaption}
\usepackage{lineno}
\usepackage{ragged2e}
\usepackage{appendix}

\RequirePackage[
datamodel=acmdatamodel,
style=acmnumeric, 
]{biblatex}
\AtBeginEnvironment{multicols}{\RaggedRight}

\newcommand\ian[1]{}
\newcommand\eunsuk[1]{}

\newcommand\flogic{SFL}

\newcommand\True{\textsc{true}}
\newcommand\False{\textsc{false}}
\newcommand\TLA{TLA\textsuperscript{+}}

\newcommand{\act}{Act\ }

\newcommand{\ag}[1]{\langle\hspace{-.9mm}\langle #1 \rangle\hspace{-.9mm}\rangle}
\newcommand{\sag}[1]{\langle #1 \rangle}

\newcommand\st{\ : \ }
\newcommand\tlastate[1]{
    \begin{array}{lll}
        #1
    \end{array}
}

\newtheorem{remark}{Remark}

\setlist[itemize]{leftmargin=*}

\begin{document}

\title{Compositional Inductive Invariant Inference via Assume-Guarantee Reasoning}

\author{Ian Dardik}
\email{idardik@andrew.cmu.edu}
\affiliation{%
  \institution{Carnegie Mellon University}
  \city{Pittsburgh}
  \state{Pennsylvania}
  \country{USA}
}

\author{Eunsuk Kang}
\email{eunsukk@andrew.cmu.edu}
\affiliation{%
  \institution{Carnegie Mellon University}
  \city{Pittsburgh}
  \state{Pennsylvania}
  \country{USA}
}

\begin{abstract}
A common technique for verifying the safety of complex systems is the \textit{inductive invariant} method.
Inductive invariants are inductive formulas that overapproximate the reachable states of a system and imply a desired safety property.
However, inductive invariants are notoriously complex, which makes \textit{inductive invariant inference} a challenging problem.
In this work, we observe that inductive invariant formulas are complex primarily because they must be closed over the transition relation of an entire system.
Therefore, we propose a new approach in which we decompose a system into components, assign an assume-guarantee contract to each component, and prove that each component fulfills its contract by inferring a \textit{local inductive invariant}.
The key advantage of local inductive invariant inference is that the local invariant need only be closed under the transition relation for the component, which is simpler than the transition relation for the entire system.
Once local invariant inference is complete, system-wide safety follows by construction because the conjunction of all local invariants becomes an inductive invariant \textit{for the entire system}.
We apply our compositional inductive invariant inference technique to two case studies, in which we provide evidence that our framework can infer invariants more efficiently than the global technique.
Our case studies also show that local inductive invariants provide modular insights about a specification that are not offered by global invariants.
\end{abstract}
\maketitle              
%
%
%


\section{Introduction}
A common approach to verifying the safety of complex systems, such as distributed protocols, is to find an \textit{inductive invariant}.
An inductive invariant is a logical formula that overapproximates the reachable state space of a system (invariance), is closed with respect to the system's transition relation (inductiveness), and also implies safety; together, these conditions imply that the system is safe.
However, inductive invariant formulas can be notoriously large and complex, even for smaller systems.
As a result, the important problem of \textit{inductive invariant inference} remains a significant research challenge.

Over the past several decades, researchers have developed principled approaches for traversing the large space of candidate inductive invariant formulas.
One common type of method is the \textit{incremental} approach \cite{bradley:2011,Goel:2021,Koenig:2020,Koenig:2022,Zohar:1995,schultz:2022indinvs} that accumulates non-inductive invariants--referred to as \textit{lemmas}--until their conjunction becomes inductive.
Yet another approach is to efficiently perform a bounded enumeration of candidate formulas \cite{Hance:2021,Yao:2021}.
Techniques such as these have made inductive invariant inference applicable to a larger range of systems, and have been used to automatically verify sophisticated distributed protocols such as Paxos \cite{Lamport:1998}.

Despite progress in taming the search space, state-of-the-art techniques still struggle to infer complex inductive invariants that are necessary for verifying large systems.
Our key observation is that inductive invariants become large and complex primarily because the formula must be closed under the transition relation of the entire system.
Therefore, instead of attempting to synthesize a global invariant for the entire system, in this paper, we propose a \textit{compositional} approach that aims to find multiple simpler inductive invariants.
Our approach involves decomposing a specification into components and separately inferring a \textit{local inductive invariant} for each component.
Local inductive invariants are only required to be closed under the transition relation of the corresponding component, rather than the entire system.
The advantage to local inductive invariant inference is that each component has a simpler transition relation than the overall system and, therefore, presents a simpler inference task.

Our proposed compositional approach is based on a novel \textit{assume-guarantee reasoning} \cite{Pnueli:1985,Giannakopoulou:2018handbook} theory.
Assume-guarantee reasoning is a verification paradigm in which each component $C$ is assigned an individual proof obligation $\sag{\alpha} C \sag{\gamma}$ called a \textit{contract}.
Intuitively, a contract establishes that $C$ guarantees the property $\gamma$ under the assumption $\alpha$.
In our theory, the goal is to infer a local inductive invariant for each component to show that it fulfills its contract.

In order to define an assume-guarantee theory that allows for proof via local inductive invariant, we consider both \emph{action-based} and \emph{state-based} behavioral semantics.
The need to consider both semantic paradigms arises because assume-guarantee contracts are a means for both \textit{composition} and \textit{proof localization}.
On the one hand, action-based semantics are convenient for composing contracts because components can be treated algebraically and composed via parallel composition.
For example, action-based composition styles from CSP \cite{Hoare1978} and the $\pi$-calculus \cite{Milner:1992} are often used in assume-guarantee reasoning for concurrent and distributed systems \cite{Magee:2006,Magee:1999}.
On the other hand, proofs via (local) inductive invariants require state-based semantics because inductive invariants are formulas over state variables.

In this paper, we introduce a two-layered assume-guarantee theory that offers the advantages of both semantic paradigms.
The top layer of the theory is action-based and intended for compositional reasoning, while the bottom layer is state-based and intended for proof via local inductive invariant.
In the two-layer approach, we decompose a system into action-based contracts in the top layer theory, which we then translate to state-based contracts in the bottom-layer theory to be proved.
The key to translating contracts from the action-based theory to the state-based theory is a novel logic language that we propose called \emph{Symbolic Fluent Linear Temporal Logic} (\emph{\flogic{}}).
\flogic{} formulas are action-based, but can be translated to a semantically equivalent state-based formula.
In other words, given an action-based contract $\sag{\alpha} C \sag{\gamma}$ where $\alpha$ and $\gamma$ are \flogic{} formulas, we can translate it to an equivalent state-based contract $\ag{A} C \ag{G}$ where $A$ and $G$ are state-based formulas.

Using our two-layered assume-guarantee theory, we present \textbf{the first divide-and-conquer approach to inductive invariant inference}.
The approach involves four steps, which are shown  in Fig.~\ref{fig:overview}.
Given a system that we wish to verify, the first step is to decompose it into a set of components.
The second step is to find \flogic{} formulas $\alpha$ and $\gamma$ for each component $C$ to create a candidate action-based contract $\sag{\alpha} C \sag{\gamma}$.
Third, for each component, we infer a local inductive invariant that proves its corresponding state-based contract $\ag{A} C \ag{G}$.
Local inductive invariant inference reduces to the usual inductive invariant inference problem, which can be accomplished using existing search-based techniques.
Therefore, our work is complementary to--rather than a replacement for--existing work in inductive invariant inference.
The fourth and final step is to take the conjunction of all local inductive invariants; by construction, the resulting formula is an inductive invariant \textit{for the entire system}.

We evaluate our technique on two case studies of distributed protocols, the first of which is the classic Two Phase Commit Protocol \cite{Gray:2006}.
The second case study is an industrial-scale protocol of a Raft \cite{Ongaro:2014} style algorithm that runs at MongoDB \cite{Mongodb}.
While the MongoDB case study is beyond the reach of all existing automatic inductive invariant inference techniques, we use our compositional inference framework to infer an inductive invariant semi-automatically.
Our results in the two case studies provide evidence that our divide-and-conquer approach can be more efficient than the global approach for inductive invariant inference.
Additionally, we show that the artifacts from compositional inference--local inductive invariants and assume-guarantee contracts--provide valuable insights into the behaviors of the system that the global approach does not offer.


In summary, our contributions are:
\begin{enumerate}
\item A two-layered assume-guarantee theory, whose contracts we prove using the novel proof by local inductive invariant method,
\item Symbolic Fluent Linear Temporal Logic (\flogic{}), a logic language for specifying action-based behavioral properties of parameterized systems,
\item A compositional verification framework driven by a divide-and-conquer inductive invariant inference algorithm,
\item Two case studies in which we show the efficacy of our verification framework and highlight the insights that the compositional approach offers about the verified system.
\end{enumerate}

The rest of the paper is structured as follows.
In Sec.~\ref{sec:background} we introduce the running example and background information for the paper.
In Sec.~\ref{sec:state-theory} we introduce the bottom layer state-based theory, followed by the top layer action-based theory in Sec.~\ref{sec:action-theory}.
Using the two layered assume-guarantee theory, we introduce our compositional inductive invariant inference framework in Sec.~\ref{sec:ii-infer}.
We present the case studies in our evaluation in Sec.~\ref{sec:evaluation}, followed by a survey of related work in Sec.~\ref{sec:related-work}.
Finally, we conclude in Sec.~\ref{sec:future-work} with a discussion of the limitations of this paper and future work.

\begin{figure}
    \centering
    \includegraphics[width=0.8\linewidth]{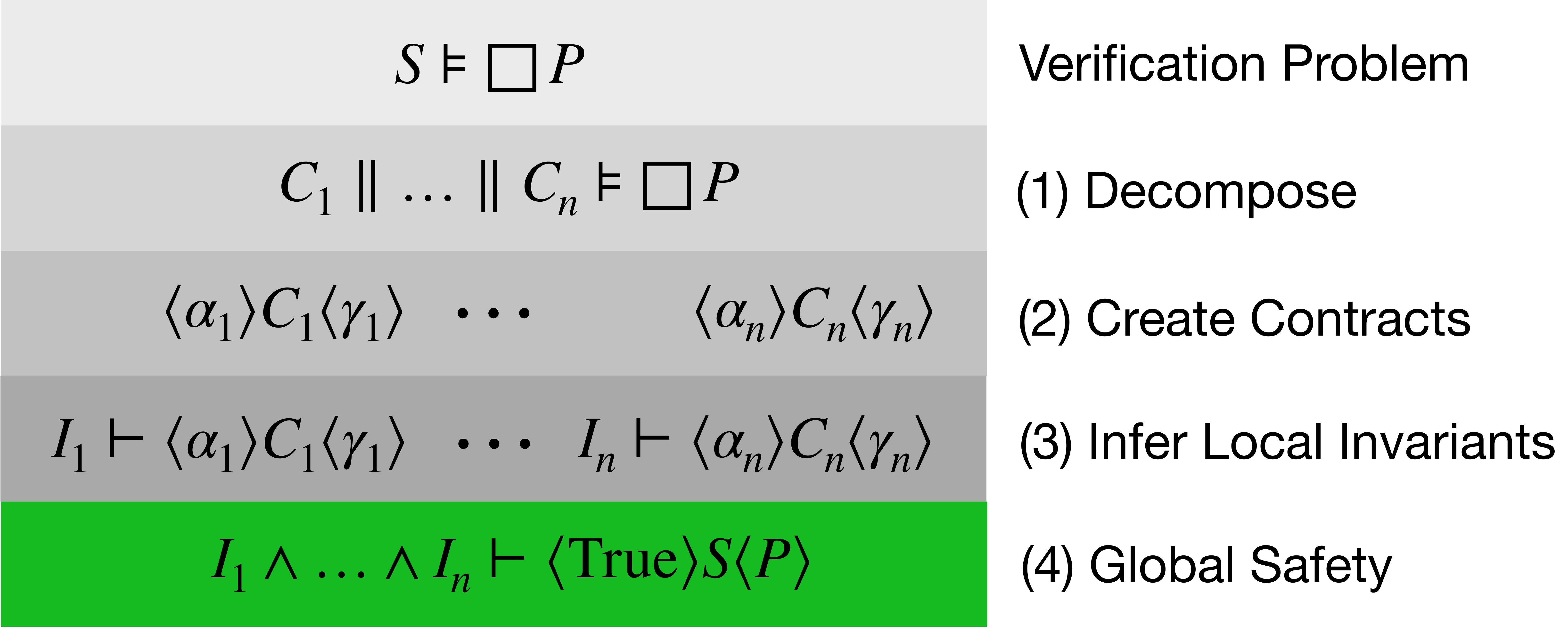}
    \caption{Overview of our compositional inductive invariant inference framework.
        Each layer in the figure represents an equivalent verification problem.
        From top to bottom, the figure shows the original verification problem, decomposition of the system $S$ into components, creation of assume-guarantee contracts for each component, local inductive invariant inference for each component, and proof of global safety by construction.
        The notation $I \vdash \sag{\alpha} C \sag{\gamma}$ indicates that $I$ is a local inductive invariant for the assume-guarantee contract.}
    \label{fig:overview}
\end{figure}
\section{Background}
\label{sec:background}
We now introduce background information, including the running example, the \TLA{} formal specification language, parameterized systems, and the inductive invariant proof method.

\paragraph{Running Example: Toy-2PC}
We now introduce a toy version of the classic Two Phase Commit protocol \cite{Gray:2006}, Toy-2PC, which we use as a running example throughout the paper.
Toy-2PC is a distributed database commit protocol in which a \textit{transaction manager} (TM) communicates with a set of \textit{resource managers} (RMs) in attempt to commit a database transaction.
The protocol unfolds over two phases, which we now describe.
In the first phase, each RM starts in the \textit{working} state as it attempts to commit the transaction; any RM that can commit a transaction sends a \textit{prepared} message to the TM.
In the second phase, the TM will issue a \textit{commit} message to each RM if they are all prepared, or an \textit{abort} message otherwise. 
The safety property for the protocol is \textit{consistency}, meaning that no two RMs should disagree as to whether a transaction was committed or aborted.
For simplicity, Toy-2PC assumes a perfect network, while the original protocol allows messages to be delayed, dropped, and reordered.

\paragraph{\TLA{}}
We will use the \TLA{} formal specification language \cite{lamport2002specifying} for the running example as well as our case studies.
\TLA{} is based on first-order logic (FOL) and temporal logic, and is often used for specifying parameterized symbolic transition systems and their temporal properties.

For example, Fig.~\ref{fig:toy-2pc} shows a specification for Toy-2PC.
In Fig.~\ref{fig:toy-2pc}, lines 1 and 2 of the specification respectively declare parameters (via the \textsc{constant} keyword) and state variables (via the \textsc{variables} keyword).
Lines 3-6 specify the initial state predicate, while the remainder of the specification encodes the actions of the symbolic transition relation.
The symbolic transition relation is the formula $\exists r \in RMs : Prepare(r) \lor Commit(r) \lor Abort(r) \lor SilentAbort(r)$; however, for brevity, we will omit the transition relation formula from the specifications shown in this paper.
In the actions, prime operators (') are used to denote the value of a variable at the following time step and the \textsc{unchanged} keyword explicitly indicates the frame condition for an action.
The notation on line 4 uses square brackets to define $rmState$ as a function, while the notation using \textsc{except} on line 11 changes one value of the function $rmState$.
Line 8 accesses the value of the $rmState$ function, while line 13 uses angle brackets to specify a tuple.

The safety property $Consistent$ for the Toy-2PC protocol can be written as an invariant in \TLA{}, shown in Fig.~\ref{fig:toy-2pc-consistent}.
An invariant is a formula of the form $\Box P$, where $\Box$ is the \textit{always} temporal operator and $P$ is a non-temporal FOL formula.
Oftentimes, we will also refer to non-temporal FOL formulas as invariants, in which case the outer $\Box$ is implied; this is the case for the formula for $Consistent$.
One can prove that the Toy-2PC protocol satisfies $\Box Consistent$ using the inductive invariant proof method, which we describe later in this section.
\begin{figure}
    \begin{subfigure}[t]{0.99\textwidth}
        {\tlatex \@x{}\moduleLeftDash\@xx{ {\MODULE} ToyTwoPhase}\moduleRightDash\@xx{}}
    \end{subfigure}
    {\tlatex \@pvspace{2.0pt}}
    \begin{subfigure}[t]{0.46\textwidth}
        \internallinenumbers{\small \input{tla/toy_2pc_left}}
    \end{subfigure}
    \begin{subfigure}[t]{0.47\textwidth}
        \internallinenumbers{\small \input{tla/toy_2pc_right}}
    \end{subfigure}
    {\tlatex \@pvspace{6.0pt}}
    {\tlatex \@x{}\bottombar\@xx{}}
    \caption{The \TLA{} specification for Toy-2PC.}
    \label{fig:toy-2pc}
\end{figure}
\begin{figure}
    {\tlatex \@xx{Consistent \.{\defeq} \.\A\, r1 ,\, r2 \.{\in} RMs \.{:} {\lnot} ( rmState [ r1 ] \.{=}\@w{abort} \.{\land}\, rmState [ r2 ] \.{=}\@w{commit} )}}
    \caption{The definition for $Consistent$, the key safety property of Toy-2PC.}
    \label{fig:toy-2pc-consistent}
\end{figure}

\paragraph{Parameterized Symbolic Transition Systems}
A parameterized symbolic transition system (PSTS) is an indexed triple $(vars,Init,Next)^p$, where $p$ is a (possibly empty) set of parameters, $vars$ is a set of state variables, and $Init$ and $Next$ are first-order logic formulas for the initial-state constraint and the symbolic transition relation respectively.
Syntactically, $Init$ may reference the parameters in $p$ and the variables in $vars$, while $Next$ may reference the parameters in $p$ and the variables in $vars \cup vars'$, where $vars' = \{ v' \mid v \in vars \}$.
For example, the \TLA{} specification for Toy-2PC in Fig.~\ref{fig:toy-2pc} defines a PSTS $(vars,Init,Next)^{\{RMs\}}$, where $vars$ and $Init$ are defined in Fig.~\ref{fig:toy-2pc}, and $Next$ is the (existentially quantified) disjunction of the actions in the specification $\exists r \in RMs : Prepare(r) \lor \dots$
Parameters are also associated with domains, e.g., $RMs$ is associated with a countably infinite set of possible resource managers.

We will use the notation $\act S$ to denote the alphabet of a PSTS $S$, $Params(S)$ for the set of parameters, and $SV(S)$ for the set of state variables.
We consider actions to be \textit{event formulas}, meaning that we will treat them as both events (e.g., in alphabets) and formulas (e.g., in the transition relation).
We also overload $Params(F)$ and $SV(F)$ to denote the set of parameters and the set of state variables for a formula $F$.
As an example, we have $\act Toy2PC = \{Prepare(rm), Commit(rm),$ $Abort(rm),$ $SilentAbort(rm)\}$, $Params(Toy2PC) = \{RMs\}$, and $SV(Toy2PC) = \{rmState, tmState, tmPrepared\}$.
Notice that, as event formulas, actions such as $Prepare(rm)$ appear in $\act Toy2PC$, but also have formula definitions in Fig.~\ref{fig:toy-2pc}.

We will consider \textit{state-based} and \textit{action-based} behaviors in this paper.
State-based and action-based behaviors are sequences of states and actions respectively, where a state is an assignment to variables and an action is an \textit{event formula} as described above.
We use subscripts to denote the $i$th state or action of a behavior.
If $\tau$ is a state-based behavior and $i \geq 0$, then $\tau_i \models P$ if and only if the formula $P$ is valid when replacing each (non-primed) variable in $P$ with the value assigned to by $\tau_i$.
If $\sigma$ is an action-based behavior, $i \geq 0$, and $s$ and $t$ are states, then $(s,t) \models \sigma_i$ if and only if the event formula $\sigma_i$ is valid when replacing each non-primed variable in $\sigma_i$ with the values assigned to by $s$, and each primed variable with the value assigned to by $t$.

A behavior satisfies a PSTS if and only if it satisfies all of its \textit{instances}.
An instance of a PSTS is a symbolic transition system $(vars,Init,Next)^{\hat{p}}$, where $\hat{p}$ is a function that maps each parameter to a subset of its domain.
For example, $\hat{p} = [RMs \mapsto \{rm_1,rm_2\}]$ corresponds to an instance for Toy-2PC.
Given a state-based behavior $\tau$ and a symbolic transition system $S = (vars,Init,Next)^{\hat{p}}$, $\tau \models S$ if and only if $\tau_0 \models Init$ and, for each $i \geq 0$, $(\tau_i,\tau_{i+1}) \models Next$.
Given an action-based behavior $\sigma$ and a symbolic transition system $S$, $\sigma \models S$ if and only if there exists a state-based behavior $\tau$ such that $\tau \models S$ and, for each $i \geq 0$, $(\tau_i,\tau_{i+1}) \models \sigma_i$.
We will use the notation $\mathcal{L}(S)$ to denote the language of $S$, which is the set of all action-based behaviors that satisfy $S$.
In a given state $s$, we say that an action $a$ is \textit{enabled} if there exists a state $t$ such that $(s,t) \models a$.
We call $S$ \textit{deterministic} if, for all states $s$, $t_1$, and $t_2$ and any action $a$, it is the case that $(s,t_1) \models a$ and $(s,t_2) \models a$ implies $t_1 = t_2$.

Let $S_i = (vars_i,Init_i,Next_i)^p$ for $i \in \{1,2\}$, where $vars_1 \cap vars_2 = \emptyset$.
We define parallel composition over PSTSs as follows: $S_1 \parallel S_2 \triangleq (vars_1 \cup vars_2, Init_1 \land Init_2, Next)^p$, where the symbolic transition relation $Next$ is defined as:
\vspace{1.8mm}
\begin{align*}
    &\bigvee_{a \in \act S} \left\{ \tlastate{
        \exists d \in D \st S_1!a \ \land \ S_2!a\quad \hfill \text{if}\ a \in \act S_1 \ \text{and} \ a \in \act S_2 \\
        \exists d \in D \st S_1!a \ \land \ S_2!vars' = S_2!vars\quad \hfill \text{if}\ a \in \act S_1 \ \text{and} \ a \notin \act S_2 \\
        \exists d \in D \st S_1!vars' = S_1!vars \ \land S_2!a\quad \hfill \text{if}\ a \notin \act S_1 \ \text{and} \ a \in \act S_2
    } \right.
\end{align*}
\vspace{0.7mm}

\noindent Here, $\exists d \in D$ quantifies over the parameters associated with each action $a$ ($d$ appears syntactically in $a$), the operator $!$ indicates the scope of the definition for a formula, and $\act S = (\act S_1) \cup (\act S_2)$.
Essentially, the symbolic transition relation synchronizes actions that are shared between components and interleaves unshared actions.
We remark that parallel composition is only defined for PSTSs that have the same parameters $p$ but do not share any state variables.
As an example, consider the decomposition of the Toy-2PC specification shown in Fig.~\ref{fig:toy-2pc-components}.
In this case, we have $Toy2PC = ToyRM \parallel ToyTM$.
We refer to $ToyRM$ and $ToyTM$ as the \textit{components} of the Toy-2PC system.
\begin{figure}
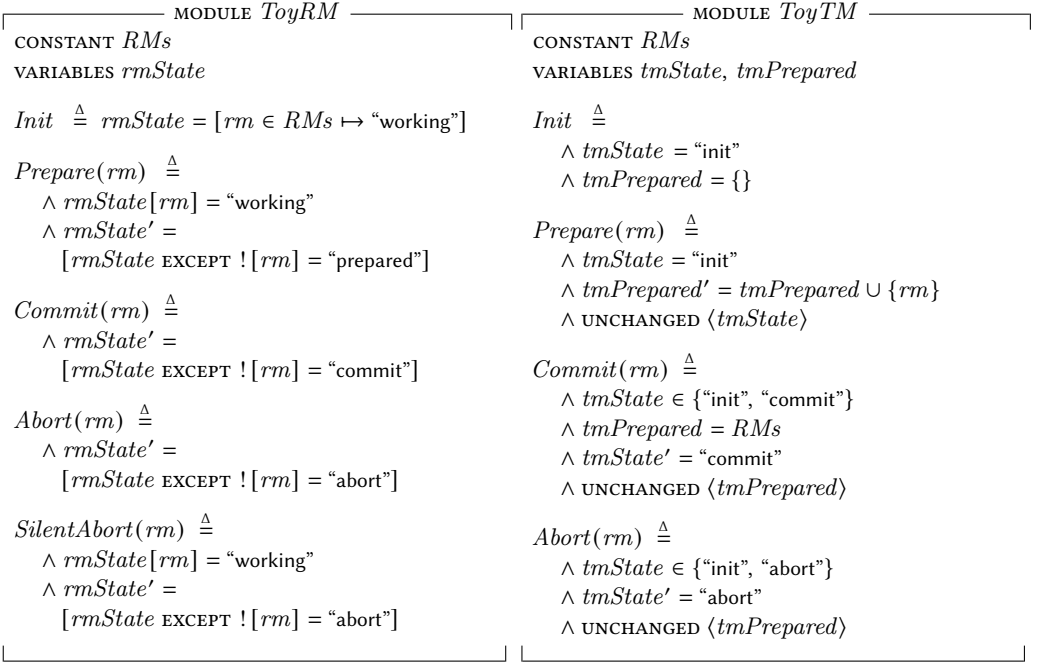

    \begin{subfigure}[t]{0.49\textwidth}
        {\small \input{tla/toy_2pc_RM}}
    \end{subfigure}
    \begin{subfigure}[t]{0.49\textwidth}
        {\small \input{tla/toy_2pc_TM}}
    \end{subfigure}
    \caption{Toy-2PC decomposed into two components, $ToyRM$ and $ToyTM$.}
    \label{fig:toy-2pc-components}
\end{figure}

\paragraph{The Inductive Invariant Proof Method}
Because PSTS are families of transition systems, typical model checking techniques, such as explicit-state enumeration and symbolic model checking, are generally insufficient for verification.
Instead, verification is usually performed using the inductive invariant proof method.
This method proves that a PSTS $(vars,Init,Next)^p$ satisfies an invariant $\Box P$ by finding a formula $Inv$, called an \textit{inductive invariant}, that satisfies the following three equations:
\begin{align}
    &Init \implies Inv \label{eqn:initiation}\\
    &Inv \land Next \implies Inv' \label{eqn:consecution}\\
    &Inv \implies P \label{eqn:safety}
\end{align}
Equation (\ref{eqn:initiation}) is called \textit{initiation} and (\ref{eqn:consecution}) is called \textit{consecution}; together, these conditions show that $Inv$ is an overapproximation of the reachable states of the transition system.
On the other hand, equation (\ref{eqn:safety}) shows that $Inv$ is an underapproximation of the invariant $\Box P$.
Together, the three conditions above imply that the transition system is safe.
The problem of \textit{inductive invariant inference} is the task of finding an inductive invariant that can be used to verify a system.
We refer to any technique that finds $Inv$ directly--without decomposing the system--as a \textit{global inductive invariant inference} technique.
\section{State-Based Assume-Guarantee Theory}
\label{sec:state-theory}
In this section, we present our \textit{state-based} assume-guarantee theory that occupies the bottom of our two-layered theory.
We begin in Sec.~\ref{sec:contracts} by defining state-based contracts, as well as the local inductive invariant proof technique for proving that a state-based contract is fulfilled.
Subsequently, in Sec.~\ref{sec:compose-contracts}, we present inference rules for composing state-based contracts and their local inductive invariants.
Although the state-based theory is not intended for composition, we use concepts from the action-based theory (e.g., parallel composition over components) to create composition rules.
Ultimately, these composition rules are not intended to be used directly for verification, but they will simplify our definitions for the composition rules in the action-based theory.
Finally, we dedicate Appendix~\ref{sec:theory-soundness} to proving that these composition rules are sound.

\subsection{State-Based Contracts and Local Inductive Invariants}
\label{sec:contracts}
We will use the notation $\ag{A} C \ag{G}$ for contracts in our state-based theory, where $G$ is the guarantee that the component $C$ must satisfy under the assumption $A$.
We base contracts on \textit{invariance}, meaning that the component may assume $A$ at all time steps, but must also guarantee $G$ at all time steps.
Therefore, the proof obligation of a contract is to satisfy the invariant $G$, while only considering those states of $C$ that satisfy $A$.
We capture this formally in the following definition.
\begin{definition}
    \label{def:ag-contract}
    Let $C = (vars,Init,Next)^p$ be a PSTS, and let $A$ and $G$ be non-temporal FOL formulas such that $SV(A) \cup SV(G) \subseteq vars$ and $Params(A) \cup Params(G) \subseteq p$.
    A contract $\ag{A} C \ag{G}$ is \textit{fulfilled} if and only if $(vars, Init \land A, Next \land A')^p \models \Box G$.
    We will denote contract fulfillment with the notation $\vdash \ag{A} C \ag{G}$.
\end{definition}
As an example, consider the $ToyRM$ component in the Toy-2PC protocol, shown in Fig.~\ref{fig:toy-2pc-components}.
If we assume--strictly for the sake of demonstration--that no resource manager will ever abort, then it is possible to prove that $ToyRM$ guarantees the safety property $Consistent$, shown in  Fig.~\ref{fig:toy-2pc-consistent}.
If we let {\tlatex \@xx{NoAbort = \A\, r \.{\in} RMs \.{:} rmState[r] \.{\neq} \@w{abort}}}, then $\vdash \ag{NoAbort} ToyRM \ag{Consistent}$.

Unfortunately, conventional model-checking techniques cannot be used to show that the above contract is fulfilled because $ToyRM$ is parameterized with an infinite domain.
Therefore, we turn our attention to a more powerful technique for showing contract fulfillment called the \textit{local inductive invariant} proof method.
In this proof method, one must find a formula $I$ that is an inductive invariant for the component under an assumption $A$, meaning that $A$ can be assumed to hold at all time steps.
We will use the notation $I \vdash \ag{A} C \ag{G}$ to mean that $I$ is a local inductive invariant that proves that the corresponding contract is fulfilled.
\begin{definition}
    \label{def:local-ind-inv}
    Let $C = (vars,Init,Next)^p$ be a PSTS, $\ag{A} C \ag{G}$ be a contract, and $I$ be a non-temporal FOL formula such that $SV(I) \subseteq vars$ and $Params(I) \subseteq p$.
    Then we write $I \vdash \ag{A} C \ag{G}$ to mean that the following three equations are valid:
    \begin{align}
        &Init \land A \implies I \\
        &I \land Next \land A \land A' \implies I' \\
        &I \implies G
    \end{align}
\end{definition}
In the Toy-2PC example, $NoAbort$ serves as a local inductive invariant to prove the contract from above, so we additionally write $NoAbort \vdash \ag{NoAbort} ToyRM \ag{Consistent}$.
Finally, we conclude this subsection by showing that finding a local inductive invariant is a sufficient condition for contract fulfillment; we provide a proof for this theorem in Appendix~\ref{apx:ii-implies-safe}.
\begin{theorem}
    \label{thm:ii-implies-safe}
    $I \vdash \ag{A} C \ag{G}$ implies $\vdash \ag{A} C \ag{G}$.
\end{theorem}

\subsection{Composing Contracts and Local Inductive Invariants}
\label{sec:compose-contracts}
In this section, we present inference rules for composing state-based contracts.
On their own, these rules can be challenging to use for verifying large systems.
Instead, these inference rules are intended to lay the groundwork for composing action-based contracts in Sec.~\ref{sec:action-theory}.
In particular, we will leverage the soundness theorems for these rules to prove the soundness of the action-based composition rules.
At the end of this subsection, we provide a discussion with the explicit reason the rules are tough to use on their own.

The inference rules for composition in this paper are based on \textit{transitivity} of invariance.
This style of contract composition resembles that of the popular and successful assume-guarantee verification theory for finite-state systems, originally introduced by Cobleigh et al. \cite{Cobleigh:2003}.
We will exclusively focus on transitive-style contract composition; however, in the future we plan to extend our theory with additional types of composition, e.g., circular composition \cite{Elkader:2018,Graf:1990,McMillan:1999,Kedar:2010} and conjunctive styles of composition as seen in the Owicki-Gries method \cite{Owicki:1976} and Concurrent Separation Logic (called the disjoint concurrency rule) \cite{OHEARN:2007}.

The idea behind our transitive composition method is to find an invariant $R$ that serves as the guarantee for one contract and the assumption of another, e.g., as seen in the rule \textsc{naive-comp} in Fig.~\ref{fig:composition-rules}.
The formula $R$ is often referred to as an \textit{assumption}, however we will refer to $R$ as a \textit{bridge formula} (or simply a \textit{bridge}), since we already reserve the term \textit{assumption} for the assumption of a contract.
In addition to composing contracts, we also provide inference rules for composing the associated local inductive invariants.
The \textsc{naive-comp} rule, for example, composes two contracts as well as their associated local inductive invariants $I_1$ and $I_2$.
The composed invariants--that is, the formula $I_1 \land I_2$--becomes a local inductive invariant for the composed contract.

Unfortunately, as the name suggests, \textsc{naive-comp} is a naive rule that cannot be used to prove fulfillment for many valid contracts.
The problem with \textsc{naive-comp} is that the components do not share state variables, and therefore the formula $R$ cannot refer to \textit{any} state variables (by Def.~\ref{def:ag-contract}).
In practice, this is a problem because local correctness arguments often rely on assumptions about the correct behaviors of their state variables.
For example, in Toy-2PC, \textsc{naive-comp} cannot be used to show that the $ToyRM$ components satisfies $Consistent$.
The reason is because $ToyRM$ only guarantees $Consistent$ if it can assume that the $ToyTM$ component operates correctly; i.e. the $ToyTM$ must only choose to commit if all resource managers have prepared, and can never issue both a commit and an abort message.
Unfortunately, this assumption cannot be expressed in terms of $ToyRM$'s state variables.
\begin{figure}
    \centering
    \begin{align*}
        &\inferrule*[lab=naive-comp]{I_1 \vdash \ag{A} C_1 \ag{R} \\ I_2 \vdash \ag{R} C_2 \ag{G}}{I_1 \land I_2 \vdash \ag{A} C_1 \parallel C_2 \ag{G}}
        \quad
        \inferrule*[lab=bridge-comp]{I_1 \vdash \ag{A} C_1 \parallel B \ag{R} \\ I_2 \vdash \ag{R} B \parallel C_2 \ag{G}}{I_1 \land I_2 \vdash \ag{A} C_1 \parallel B \parallel C_2 \ag{G}}
    \end{align*}
    \begin{align*}
        &\inferrule*[lab=aux-comp]{I_1 \vdash \ag{A} C_1 \parallel B \ag{R} \\ I_2 \vdash \ag{R} B \parallel C_2 \ag{G} \\ Aux\ B}{\vdash \ag{A} C_1 \parallel C_2 \ag{G}}
    \end{align*}
    \caption{Inference rules for composing contracts and inductive invariants in the state-based theory.}
    \label{fig:composition-rules}
\end{figure}

To address the issue identified above, we propose a \textit{bridge component} that acts as an intermediary between the two components.
The key idea--shown by rule \textsc{bridge-comp} in Fig.~\ref{fig:composition-rules}--is that the components share the bridge component $B$, so it becomes possible to write a bridge assumption $R$ over the ``shared'' state variables in $B$.
While the bridge component can be an existing component of the system, we propose creating bridge components exclusively with \textit{auxiliary variables}.
Auxiliary variables were originally proposed by Owicki and Gries to make their inference system complete \cite{Owicki:1976aux}.
Similarly, auxiliary variables make our theory more widely applicable, though not necessarily complete.
We lift the notion of auxiliary variables to \textit{auxiliary components} in the definition below.
\begin{definition}
    \label{def:aux-component}
    A component $B$ is an \textit{auxiliary component} if, for any contract $\ag{A} C \ag{G}$, we have $\vdash \ag{A} C \parallel B \ag{G}$ implies $\vdash \ag{A} C \ag{G}$.
    When $B$ is an auxiliary component we write $Aux\ B$.
\end{definition}
In the case that $Aux\ B$, we will refer to a pair $(B,R)$ that is used with the \textsc{bridge-comp} rule as a \textit{bridge pair}.
Using an auxiliary bridge component with the \textsc{bridge-comp} rule is also sound, in the sense that it proves $\vdash \ag{A} C_1 \parallel C_2 \ag{G}$.
This fact follows from the \textsc{aux-comp} inference rule in Fig.~\ref{fig:composition-rules}.
Notice that, in general, the inductive invariant $I_1 \land I_2$ in the conclusion of the \textsc{bridge-comp} rule will contain auxiliary variables from the bridge component.
Therefore, we include \textsc{aux-comp} as a separate inference rule because $I_1 \land I_2$ may not be well-defined over $C_1 \parallel C_2$.
Despite this technicality, the formula $I_1 \land I_2$ is nevertheless an inductive invariant that proves $\vdash \ag{A} C_1 \parallel C_2 \ag{G}$ when using an auxiliary bridge component.

\paragraph{Toy-2PC example.}
\begin{figure}
    \begin{subfigure}[t]{0.95\textwidth}
        {\tlatex \@x{}\moduleLeftDash\@xx{ {\MODULE} ToyB}\moduleRightDash\@xx{}}
    \end{subfigure}
    {\tlatex \@pvspace{2.0pt}}
    \begin{subfigure}[t]{0.46\textwidth}
        {\small \input{tla/toy_2pc_B_left}}
    \end{subfigure}
    \begin{subfigure}[t]{0.47\textwidth}
        {\small \input{tla/toy_2pc_B_right}}
    \end{subfigure}
    \begin{subfigure}[t]{0.95\textwidth}
        {\tlatex \@pvspace{6.0pt}}
        {\tlatex \@x{}\bottombar\@xx{}}
    \end{subfigure}
    \caption{A bridge component for Toy-2PC.}
    \label{fig:toy-2pc-bridge}
\end{figure}
\begin{figure}
    \begin{subfigure}[t]{0.80\textwidth}
        {\tlatex \input{tla/toy_2pc_R}}
    \end{subfigure}
    \caption{A bridge formula and local inductive invariants for the components of Toy-2PC.}
    \label{fig:toy-2pc-defs}
\end{figure}
Using our state-based theory, we will compositionally infer an inductive invariant for the Toy-2PC example. 
We begin by creating a bridge component $ToyB$ (Fig.~\ref{fig:toy-2pc-bridge}) with three auxiliary variables $oncePrepare$, $onceCommit$, and $onceAbort$.
These variables respectively represent whether a $Prepare$, $Commit$, or $Abort$ action has happened at least once in the past.
We define the bridge formula $ToyR$ (Fig.~\ref{fig:toy-2pc-defs}) over the auxiliary variables in $ToyB$.
This bridge formula formally encodes our earlier intuition that, in order to verify consistency, $ToyRM$ must assume that the $ToyTM$ will only choose to commit if all resource managers have prepared, and also that $ToyTM$ will never issue both a commit and an abort message.

Next, we create a local inductive invariant $I_{RM}$ (Fig.~\ref{fig:toy-2pc-defs}) for $ToyRM \parallel ToyB$, i.e. we have $I_{RM} \vdash \ag{ToyR} ToyRM \parallel ToyB \ag{Consistent}$.
Similarly, we create a local inductive invariant $I_{TM}$ (Fig.~\ref{fig:toy-2pc-defs}) for $ToyB \parallel ToyTM$, i.e. $I_{TM} \vdash \ag{\True} ToyB \parallel ToyTM \ag{ToyR}$.
The rule \textsc{aux-comp} allows us to conclude that the entire protocol (without assumptions) is safe, since $Toy2PC = ToyRM \parallel ToyTM$.
Furthermore, the inference rule \textsc{bridge-comp} allows us to infer that $I_{RM} \land I_{TM}$ is an inductive invariant for $ToyRM \parallel ToyB \parallel ToyTM$.
Because $ToyB$ is an auxiliary component, $I_{RM} \land I_{TM}$ is an inductive invariant \textit{for the entire system}.

\paragraph{Purpose of the state-based composition rules.}
The inference rules from Fig.~\ref{fig:composition-rules} can be used to show system-wide correctness, but require inventing a bridge component.
Unfortunately, bridge components can be complex--especially for large systems--which can ultimately cause these composition rules to be difficult to use in practice.
In Sec.~\ref{sec:action-theory}, we will solve this problem by introducing simpler composition rules in the action-based theory that do not require inventing a bridge component.
Therefore, the purpose of the composition inference rules from Fig.~\ref{fig:composition-rules} is to build upon them in the action-based theory.

\section{Action-Based Assume-Guarantee Theory}
\label{sec:action-theory}
In this section, we present the action-based theory that sits at the top of our two-layered assume-guarantee theory.
The action-based theory has inference rules for composition that are well suited for verification, which we will use in our compositional framework in Sec.~\ref{sec:ii-infer}.
In particular, the inference rules for the action-based theory do not require inventing a bridge component.
In the state-based theory, bridge components were a necessary intermediary for bridge formulas, since components do not share state variables.
However, components share actions, which is the basis for the simplicity of composition in the action-based theory.
In the action-based theory, assumptions, guarantees, and bridge formulas are encoded as \textit{action invariants}, which are formulas with action-based semantics.
Bridge formulas are required to reference the actions from the shared alphabet of components whose contracts are being composed, which makes an intermediary, such as a bridge component, unnecessary.

The advantages of using action-based formalisms for composition are well-known, which is the reason that assume-guarantee verification frameworks often use labeled transition systems for assumptions, guarantees, and bridges \cite{alur:2005,Chen:2009,Cobleigh:2003,Gupta:2007,Nam:2008}.
However, using action invariants in our contracts has the unique advantage that it allows our contracts to be translated to state-based contracts, which can then be proved with the local inductive invariant method.
Allowing our contracts to be proved via local inductive invariants is a key advantage because it allows us to compositionally verify parameterized systems, which the assume-guarantee theories for labeled transition systems cannot, in general, be used for.

We will begin in Sec.~\ref{sec:sfl-lang} by introducing \flogic{}, an action-based logic language that we will use to define action invariants.
In Sec.~\ref{sec:sfl-translation}, we show how to translate action invariants to state-based formalisms.
In Sec.~\ref{sec:ag-theory-sfl}, we introduce our action-based assume-guarantee theory, in which we define contracts, local inductive invariants for action-based contracts, and inference rules for composing action-based contracts.
Finally, in Sec.~\ref{sec:ag-hybrid}, we introduce a second type of action-based contract we call a \textit{hybrid contract} that allows the action-based theory to verify state-based properties.

\subsection{Symbolic Fluent LTL}
\label{sec:sfl-lang}
We now introduce a novel logic language called \textit{Symbolic Fluent Linear Temporal Logic} (\flogic{}).
\flogic{} is based on Fluent Linear Temporal Logic (FLTL) \cite{Giannakopoulou:2003}, a logic language for specifying temporal properties over \textit{actions}.
\flogic{} provides two extensions to FLTL that make it particularly well suited for specifying assumptions and guarantees for parameterized systems: \textit{symbolic fluents} and quantifiers.
We will describe these extensions along with the syntax and semantics for the new logic language.
At the end of this subsection, we will use the \flogic{} language to identify a class of formulas called \textit{action invariants} that we will use to define action-based contracts in Sec.~\ref{sec:ag-theory-sfl}.

\paragraph{\flogic{} syntax.}
\begin{figure}
    \begin{subfigure}[t]{0.31\textwidth}
        \vspace{2mm}
        {\small \input{tla/toy_2pc_oncePrepare}}
        \caption{Definition for $oncePrepareT$.}
        \label{fig:toy-2pc-oncePrepare}
    \end{subfigure}
    \hfill
    \begin{subfigure}[t]{0.31\textwidth}
        \vspace{2mm}
        {\small \input{tla/toy_2pc_onceCommit}}
        \caption{Definition for $onceCommitT$.}
        \label{fig:toy-2pc-onceCommit}
    \end{subfigure}
    \hfill
    \begin{subfigure}[t]{0.31\textwidth}
        \vspace{2mm}
        {\small \input{tla/toy_2pc_onceAbort}}
        \caption{Definition for $onceAbortT$.}
        \label{fig:toy-2pc-onceAbort}
    \end{subfigure}
    \begin{subfigure}[t]{0.46\textwidth}
        \begin{align*}
            oncePrepare &= \left(\langle RMs \rangle, x, oncePrepareT \right) \\
            onceCommit &= \left(\langle RMs \rangle, y, onceCommitT \right) \\
            onceAbort &= \left(\langle RMs \rangle, z, onceAbortT \right)
        \end{align*}
        \caption{Fluent definitions for $\rho_1$ and $\rho_2$.}
        \label{fig:toy-2pc-fluents}
    \end{subfigure}
    \hfill
    \begin{subfigure}[t]{0.46\textwidth}
        \begin{align*}
            \sigma^1 &= \big(Prepare(r_1)\big) \\
            \sigma^2 &= \big(Prepare(r_1), Prepare(r_2), Commit(r_2)\big)
        \end{align*}
        \caption{Example finite behaviors for the instance $RMs = \{r_1,r_2\}$.}
        \label{fig:toy-2pc-traces}
    \end{subfigure}
    \begin{subfigure}[t]{0.99\textwidth}
        \begin{align*}
            &\rho_1 \triangleq (\exists r \in RMs : onceCommit(r)) \implies (\forall r \in RMs : oncePrepare(r))\\
            &\rho_2 \triangleq (\exists r \in RMs : onceAbort(r)) \implies (\forall r \in RMs : \neg onceCommit(r))
        \end{align*}
        \caption{Example \flogic{} formulas.}
        \label{fig:toy-2pc-sfl-formula}
    \end{subfigure}
    \caption{Definitions for specifying and verifying Toy-2PC compositionally using \flogic{}.}
\end{figure}
Symbolic fluents are an extension of the non-symbolic fluents from FLTL.
Non-symbolic fluents are action-based propositions in FLTL that are syntactically analogous to the state-based atomic propositions of Linear Temporal Logic (LTL) \cite{Pnueli:1977}.
Semantically, however, fluents represent booleans that may turn on or off depending on each action taken.
The set of non-symbolic fluents is fixed in each FLTL formula, which makes the language well-suited for specifying properties of systems that have a fixed alphabet, e.g., labeled transition systems.
In a parameterized system, however, the alphabet generally depends on the parameters, which makes FLTL ill-suited for specifying properties of such systems.
This motivates \textit{symbolic fluents} that accept arguments and semantically correspond to boolean-valued functions.
Note that we specifically use the term \textit{argument} for symbolic fluents to differentiate from the \textit{parameters} of a protocol.
For the remainder of this paper, we will use the terms \textit{fluent} and \textit{symbolic fluent} synonymously.
The formal definition of a fluent is as follows.
\begin{definition}
    \label{def:sym-fluent}
    A symbolic fluent is a triple $(s,v,T)$, where $s$ is a tuple of parameters that specify the type of each argument to the fluent, $v$ is a state variable, and $T$ is a PSTS with the following four requirements:
    \begin{enumerate}
        \item $T$ has exactly one state variable $v$.
        \item $T$ is deterministic.
        \item The type of $v$ must always be a boolean-valued function whose arguments have the types given by $s$.
        \item Every action in $\act T$ is enabled in every state of $T$.
    \end{enumerate}
    We also define the following notation for the alphabet of a symbolic fluent: $\act(s,v,T) = \act T$.
\end{definition}
Returning to the Toy-2PC example, consider the fluent $oncePrepare$ in Fig.~\ref{fig:toy-2pc-fluents}.
The fluent accepts one argument with type $RMs$ and refers to a state variable $x$ whose value changes according to the transition system $oncePrepareT$ from Fig.~\ref{fig:toy-2pc-oncePrepare}.
Intuitively, the fluent represents whether a $Prepare(rm)$ action has occurred at least once in the past, which can be seen from the definition of $oncePrepareT$.
We will use the syntax $oncePrepare(rm)$ to access the value of the fluent at a given time step, where $rm$ is an argument to the fluent.
For example, after execution of the finite behavior $\sigma^1$ from Fig.~\ref{fig:toy-2pc-traces}, $oncePrepare(r_1)$ will evaluate to $\True$, while $oncePrepare(r_2)$ will evaluate to $\False$.
Because fluents in \flogic{} accept arguments, we can introduce quantifiers into the language for the purpose of passing quantified variables as arguments to fluents.
Based on these extensions, we now define the syntax for \flogic{} formulas.
\begin{definition}
    A Symbolic Fluent Linear Temporal Logic (\flogic{}) formula has the syntax in Fig.~\ref{fig:sfl-syntax}, must be absent of free variables (arguments to fluents), and fluents may not share state variables.
    For an \flogic{} formula $\phi$, $fl(\phi)$ refers to the set of all symbolic fluents in $\phi$ and $Params(\phi)$ refers to the set of all parameters (each $D$ from Fig.~\ref{fig:sfl-syntax}) that are quantified over in $\phi$.
    In this paper, we will use Greek letters, namely $\phi,\psi,\alpha,\rho,$ and $\gamma$, to represent \flogic{} formulas.
\end{definition}
\begin{remark}
    We include standard logical connectives (e.g., $\land$ and $\Longrightarrow$) and temporal logic operators (e.g., $\Box$ and $\lozenge$) in \flogic{}, defined with the usual syntactic sugar.
    For example, the ``always'' operator can be defined in \flogic{} as $\Box \phi = \neg(\True\ \textbf{U} \neg\phi)$.
\end{remark}
\begin{figure}
    \centering
    $$\flogic{} ::= \ \flogic{}\ \textbf{U}\ \flogic{} \mid \textbf{X}\ \flogic{} \mid \exists x \in D : \flogic{} \mid \ \neg\flogic{} \mid \flogic{}\lor\flogic{} \mid f(ARG)$$
    $$ARG ::= \ x \mid x , ARG$$
    \caption{BNF grammar for the \flogic{} Syntax.
        \textbf{U} and \textbf{X} are temporal operators, while $\exists$, $\neg$, and $\lor$ are first-order logic connectives.
        $f(ARG)$ indicates a symbolic fluent $f$ with one or more parameters.
        In $ARG$, $x$ represents a quantified variable passed as an argument to a symbolic fluent.}
    \label{fig:sfl-syntax}
\end{figure}
For example, consider the \flogic{} formula $\Box\rho_1$, where $\rho_1$ is defined in Fig.~\ref{fig:toy-2pc-sfl-formula}.
This formula is an invariant that quantifies over $RMs$ and has two symbolic fluents, $oncePrepare$ and $onceCommit$, that are defined in Fig.~\ref{fig:toy-2pc-fluents}.
The $onceCommit$ fluent is defined similarly to the $oncePrepare$ fluent described above, the main difference being that $onceCommit$ represents whether a $Commit(rm)$ action has occurred at least once (see the definition for $onceCommitT$ in Fig.~\ref{fig:toy-2pc-onceCommit}).
Intuitively, the formula $\rho_1$ specifies that at all times, if any resource manager commits, then it must be the case that all resource managers have previously prepared.

\paragraph{\flogic{} semantics.}
\begin{figure}
    \centering
    \begin{align*}
        &E \vdash \sigma,i \models \phi\textbf{U}\psi && \text{iff} \quad \exists j \in \mathbb{N} : (j \geq i) \land (E \vdash \sigma,j \models \psi) \ \land \\
        & && \quad\quad (\forall k \in \mathbb{N} : i \leq k < j \implies E \vdash \sigma,k \models \phi) \\
        &E \vdash \sigma,i \models \textbf{X}\phi && \text{iff} \quad E \vdash \sigma,i+1 \models \phi \\
        &E \vdash \sigma,i \models \exists x \in D : \phi && \text{iff} \quad \exists y \in D : E[x \mapsto y] \vdash \sigma,i \models \phi \\
        &E \vdash \sigma,i \models \neg \phi && \text{iff} \quad \neg (E \vdash \sigma,i \models \phi) \\
        &E \vdash \sigma,i \models \phi \lor \psi && \text{iff} \quad (E \vdash \sigma,i \models \phi) \lor (E \vdash \sigma,i \models \psi) \\
        &E \vdash \sigma,i \models f(r) && \text{iff} \quad (f|\sigma_0\dots\sigma_i)[r[E]]
    \end{align*}
    \caption{\flogic{} semantics.}
    \label{fig:sfl-semantics}
\end{figure}
We now turn our attention to defining the formal semantics of \flogic{}.
We use a ternary operator of the form $E \vdash \sigma,i \models \phi$ to express that a (possibly infinite) action-based behavior $\sigma$ satisfies the \flogic{} formula $\phi$ at index $i$ with environment $E$.
An environment is an assignment of variables to values, and we use the notation $r[E]$ to mean the result of replacing all variables in $r$ with their corresponding value in $E$.
Additionally, we use the short hand $\sigma \models \phi$ to mean $[] \vdash \sigma,0 \models \phi$, where $[]$ is the empty environment.
We also define the language of an \flogic{} formula $\phi$ as the set of all its satisfying behaviors: $\mathcal{L}(\phi) = \{ \sigma \mid \sigma \models \phi \}$.

We present the formal semantics for \flogic{} in Fig.~\ref{fig:sfl-semantics}.
The semantics for \flogic{} are standard, with the exception of symbolic fluent expressions which we now describe.
Intuitively, the meaning of a symbolic fluent $(s,v,T)$ is the boolean-valued function $v$ after the execution of a finite behavior $\sigma$ of actions in the transition system $T$.
We now provide a notation for accessing the value of a fluent at a given time step.
\begin{definition}
    \label{def:fluent-bar}
    Given a fluent $f = (s,v,T)$ and a finite action-based behavior $\sigma = \sigma_0\sigma_1 \dots \sigma_i$, we let $f|\sigma$ denote the boolean-valued function $v$ that is the result of executing $\sigma$ in the transition system $T$.
    Note that any action in $\sigma$ that is not in $\act T$ leaves the state variable $v$ unchanged.
    We remark that this notation is well-defined because $\sigma$ is finite, $T$ is deterministic, and all actions in $T$ are always enabled.
\end{definition}
Fig.~\ref{fig:sfl-semantics} uses the notation from Def.~\ref{def:fluent-bar} to define the semantics of a symbolic fluent expression $f(r)$, where $f$ is a fluent and $r$ are the arguments passed to the fluent.
In $(f|\sigma_0\dots\sigma_i)[r[E]]$ from the definition in Fig.~\ref{fig:sfl-semantics}, the arguments $r$ contain quantified variables, and therefore $r[E]$ are concrete arguments.
Essentially, the semantics of a fluent is the (boolean) value of the function $(f|\sigma_0\dots\sigma_i)$ given the concrete arguments $r[E]$.

For example, consider $oncePrepare$ from Fig.~\ref{fig:toy-2pc-fluents} as well as the finite behaviors from Fig.~\ref{fig:toy-2pc-traces}.
In this case, $oncePrepare|\sigma^1 = [r_1 \mapsto \True, r_2 \mapsto \False]$ and $onceCommit|\sigma^2 = [r_1 \mapsto \False, r_2 \mapsto \True]$.
Alternatively, we can also write $(oncePrepare|\sigma^1)[r_1] = \True$, $(oncePrepare|\sigma^1)[r_2] = \False$, etc. 
We also point out that we have both $\sigma^2 \models \Box\rho_1$ and $\sigma^2 \models \Box\rho_2$.

\paragraph{Action invariants.}
In the remainder of this paper, we will focus on a class of \flogic{} formulas called \textit{action invariants}.
An action invariant is a formula of the form $\Box\phi$, where $\phi$ is a non-temporal \flogic{} formula.
For example, $\Box \rho_1$ and $\Box \rho_2$ are both action invariants.
In Sec.~\ref{sec:sfl-translation}, we will show that action invariants can be translated to state-based formalisms.
Subsequently, in Sec.~\ref{sec:ag-theory-sfl}, we will use action invariants for the assumptions and guarantees of the contracts in our action-based assume-guarantee theory.


\subsection{Translating Action Invariants to State-Based Formalisms}
\label{sec:sfl-translation}
In this section, we show that action invariants have two state-based counterparts: bridge pairs and transition systems (PSTSs).
When we introduce our action-based assume-guarantee theory in the following section (Sec.~\ref{sec:ag-theory-sfl}), the translation from action invariants to bridge pairs will allow us to translate action-based contracts to state-based contracts. 
On the other hand, the translation from action invariants to transition systems lets us treat action invariants as algebraic processes, which will help us to define action-based contracts.

\paragraph{Translation to bridge pairs.}
We now show that action invariants correspond to bridge pairs, which are state-based formalisms that we introduced in Sec.~\ref{sec:compose-contracts}.
The translations are given by two mappings $\mathcal{B}$ (bridge components) and $\mathcal{R}$ (bridge formulas), which we now define.
\begin{definition}
    \label{def:b-r-maps}
    Let $\phi$ be an \flogic{} formula, then $\mathcal{B}(\phi) = \hspace{1.5mm} \parallel\hspace{-1mm}\{ T \mid (s,v,T) \in fl(\phi) \}$, where $fl(\phi)$ is the set of all fluents in the formula $\phi$ and the $\parallel$ operator here indicates the parallel composition of all transition systems in the given set.
    We provide the formal definition for $\mathcal{R}$ in Fig.~\ref{fig:sfl-f-map}.
\end{definition}
In essence, $\mathcal{B}$ is the parallel composition of all fluents in an \flogic{} formula, while $\mathcal{R}$ is a syntactic transformation that replaces each fluent with its underlying state variable.
These maps can be seen as a decoupling of an \flogic{} formula into two parts, where $\mathcal{B}$ represents the dynamics of the fluents and $\mathcal{R}$ represents the invariant itself.
\begin{figure}
    \centering
    \begin{align*}
        &\mathcal{R}(\phi\textbf{U}\psi) = \mathcal{R}(\phi)\textbf{U}\mathcal{R}(\psi) && \mathcal{R}(\textbf{X}\phi) = \textbf{X}\mathcal{R}(\phi) && \mathcal{R}(\exists x \in D : \phi) = \exists x \in D : \mathcal{R}(\phi) \\
        &\mathcal{R}(\neg \phi) = \neg\mathcal{R}(\phi) && \mathcal{R}(\phi \lor \psi) = \mathcal{R}(\phi) \lor \mathcal{R}(\psi) && \mathcal{R}\big((s,v,T)(r)\big) = v[r]
    \end{align*}
    \caption{Definition of $\mathcal{R}$, a syntactic transformation of an \flogic{} formula to a state-based one.}
    \label{fig:sfl-f-map}
\end{figure}

In the Toy-2PC example, we can translate the action invariant $\Box(\rho_1 \land \rho_2)$ (Fig.~\ref{fig:toy-2pc-sfl-formula}) to the corresponding state-based formalisms.
Up to a variable renaming, we have $\mathcal{B}(\Box(\rho_1 \land \rho_2)) = ToyB$ (Fig.~\ref{fig:toy-2pc-bridge}) and $\mathcal{R}(\Box(\rho_1 \land \rho_2)) = \Box ToyR$ (Fig.~\ref{fig:toy-2pc-defs}).
In this example, $\mathcal{B}$ maps an action invariant to an auxiliary component.
The following theorem shows that $\mathcal{B}$ will, in general, map action invariants to auxiliary components.
This result formally demonstrates that action invariants correspond to bridge pairs; we provide a proof in Appendix~\ref{apx:b-is-aux}.
\begin{theorem}
    \label{thm:b-is-aux}
    Let $\Box\phi$ be an action invariant, then $\mathcal{B}(\Box\phi)$ is an auxiliary component.
\end{theorem}

\paragraph{Translation to a transition system.}
We now show that action invariants correspond to transition systems.
We provide this connection formally with the mapping $\mathcal{T}$, which we define in terms of the previously introduced maps $\mathcal{B}$ and $\mathcal{R}$.
\begin{definition}
    \label{def:sfl-st}
    Let $\Box\phi$ be an action invariant and also let $\mathcal{B}(\Box\phi) = (vars,Init,Next)^p$.
    Then, $\mathcal{T}(\Box\phi) = (vars, Init \land \mathcal{R}(\phi), Next \land \mathcal{R}(\phi)')^p$.
\end{definition}
Intuitively, $\mathcal{T}(\Box\phi)$ is a transition system version of the action invariant $\Box\phi$.
We will show formally that these two are semantically equivalent (Thm.~\ref{thm:sfl-aut-equiv}) after a brief example.
Returning to the Toy-2PC example, $\mathcal{T}(\Box(\rho_1 \land \rho_2))$ is a transition system that represents $ToyB$ restricted to the behaviors of $\Box ToyR$.
This intuition helps to see that $\mathcal{T}(\Box(\rho_1 \land \rho_2)) \parallel ToyRM \models \Box Consistent$ and also that $\mathcal{L}(ToyTM) \subseteq \mathcal{L}(\mathcal{T}(\Box(\rho_1 \land \rho_2)))$.
In this example, $\mathcal{T}(\Box(\rho_1 \land \rho_2))$ appears to take the place of a bridge formula, acting as an assumption for $ToyRM$ to satisfy $Consistent$ and as a guarantee for $ToyTM$.
Indeed, we will make this intuition precise in the following section (Sec.~\ref{sec:ag-theory-sfl}) by defining action-based contracts in terms of $\mathcal{T}$.

We now present the following theorem that shows that action invariants are semantically equivalent to their transition system counterpart given by $\mathcal{T}$.
We provide a proof for this theorem in Appendix~\ref{apx:proof-sfl-aut-equiv}.
\begin{theorem}
    \label{thm:sfl-aut-equiv}
    For any action invariant $\Box\phi$, we have $\mathcal{L}(\Box\phi) = \mathcal{L}(\mathcal{T}(\Box\phi))$.
\end{theorem}

\subsection{Assume-Guarantee Theory with Action-Based Contracts}
\label{sec:ag-theory-sfl}
We now present our action-based assume-guarantee theory.
We begin by defining contracts and the local inductive invariant proof method in this theory. 
Subsequently, we provide inference rules for composing contracts and their associated invariants.

\paragraph{Action-based contracts and local inductive invariants.}
We will use the notation $\sag{\alpha} C \sag{\gamma}$ for contracts in our action-based theory, where $\gamma$ is the guarantee that $C$ must satisfy under the assumption $\alpha$.
The meaning of a contract is that the language of the component $C$, when restricted to the assumption $\alpha$, must be contained in the language of the guarantee $\gamma$.
In typical assume-guarantee verification paradigms, restricting the component to the assumption is achieved algebraically, i.e. by taking the parallel composition of the assumption and the component.
Similarly, we will restrict the component to the assumption by taking the parallel composition of $C$ and $\mathcal{T}(\Box\alpha)$, the latter being semantically equivalent to the assumption by Thm.~\ref{thm:sfl-aut-equiv}.
We now define action-based contracts from this intuition.
\begin{definition}
    \label{def:sfl-ag-contract}
    Let $C = (vars,Init,Next)^p$ be a PSTS.
    Also let $\alpha$ and $\gamma$ be non-temporal \flogic{} formulas such that $Params(\alpha) \cup Params(\gamma) \subseteq p$ and $\act \alpha \cup \act \gamma \subseteq \act C$.
    Then a contract $\sag{\alpha} C \sag{\gamma}$ is \textit{fulfilled} if and only if $\mathcal{L}(\mathcal{T}(\Box\alpha) \parallel C) \subseteq \mathcal{L}(\Box\gamma)$.
    We will denote contract fulfillment with the notation $\vdash \sag{\alpha} C \sag{\gamma}$.
\end{definition}
For example, $\sag{\True} ToyTM \sag{\rho_1 \land \rho_2}$ is an action-based contract.
Proving that the contract is fulfilled, however, requires the local inductive invariant method.
We can use the local inductive invariant method by translating the contract to a state-based contract, which we accomplish by translating the assumption and guarantee using the maps $\mathcal{B}$ and $\mathcal{R}$ from the prior section (Sec.~\ref{sec:sfl-translation}).
For an action-based contract $\sag{\alpha} C \sag{\gamma}$, the corresponding state-based contract is $\ag{\mathcal{R}(\alpha)} \mathcal{B}(\Box\alpha) \parallel C \parallel \mathcal{B}(\Box\gamma) \ag{\mathcal{R}(\gamma)}$.
We will formally show that this choice of translation is appropriate below when we prove Thm.~\ref{thm:action-ii-sound}.
Based on this translation, we provide the following definition for showing that a contract is fulfilled by a local inductive invariant.
\begin{definition}
    \label{def:sfl-local-ii}
    Let $\alpha$ and $\gamma$ be non-temporal \flogic{} formulas.
    We define $I \vdash \sag{\alpha} C \sag{\gamma}$ to be fulfilled if and only if $I \vdash \ag{\mathcal{R}(\alpha)} \mathcal{B}(\Box\alpha) \parallel C \parallel \mathcal{B}(\Box\gamma) \ag{\mathcal{R}(\gamma)}$.
    Note that, for any non-temporal \flogic{} formula $\phi$, $\mathcal{B}(\Box\phi) = \mathcal{B}(\phi)$, and hence we will usually write $\mathcal{B}(\phi)$ below instead of $\mathcal{B}(\Box\phi)$ for brevity.
\end{definition}
The contract from the example above can be proved using the invariant $I_{TM}$ from Fig.~\ref{fig:toy-2pc-defs}.
In other words, we have $I_{TM} \vdash \sag{\True} ToyTM \sag{\rho_1 \land \rho_2}$.
The following theorem that shows that the local inductive invariant method is sufficient for showing that action-based contracts are fulfilled; we provide proof in Appendix~\ref{apx:action-ii-sound}.
\begin{theorem}
    \label{thm:action-ii-sound}
    $I \vdash \sag{\alpha} C \sag{\gamma}$ implies $\vdash \sag{\alpha} C \sag{\gamma}$
\end{theorem}

\paragraph{Composing contracts and local inductive invariants.}
We now present inference rules for composing action-based contracts based on transitivity of action invariance.
The idea is that composition depends on finding an action invariant $\rho$--a bridge formula--that acts as the assumption for one contract and the guarantee for the other.
\begin{figure}
    \centering
    \begin{align*}
        &\inferrule*[lab=sfl-comp]{I_1 \vdash \sag{\alpha} C_1 \sag{\rho} \\ I_2 \vdash \sag{\rho} C_2 \sag{\gamma}}{I_1 \land I_2 \vdash \sag{\alpha} C_1 \parallel \mathcal{B}(\rho) \parallel C_2 \sag{\gamma}}
        \quad
        \inferrule*[lab=sfl-safe]{I_1 \vdash \sag{\alpha} C_1 \sag{\rho} \\ I_2 \vdash \sag{\rho} C_2 \sag{\gamma}}{\vdash \sag{\alpha} C_1 \parallel C_2 \sag{\gamma}}
    \end{align*}
    \caption{Assume-guarantee inference rules for transitive composition with an \flogic{} bridge formula $\rho$.}
    \label{fig:sfl-composition-rules}
\end{figure}

Based on this intuition, we present the inference rule \textsc{sfl-comp} in Fig.~\ref{fig:sfl-composition-rules} that composes contracts and their associated local inductive invariants.
Notice that, in the conclusion of \textsc{sfl-comp}, we include $\mathcal{B}(\rho)$ to ensure that $I_1 \land I_2$ is well-defined, given that $I_1$ and $I_2$ may reference the (auxiliary) variables in $\mathcal{B}(\rho)$.
We also point out that, by Def.~\ref{def:ag-contract}, the alphabet of a bridge formula must be contained in the shared alphabet of the two components whose contracts are being composed.
Because the actions of a bridge formula are shared between the two components, the action-based inference rules avoid an intermediary such as a bridge component.
Furthermore, the inference rule \textsc{sfl-safe} in Fig.~\ref{fig:sfl-composition-rules} shows that \textsc{sfl-comp} is sound, in the sense that it proves the composition of the components to be safe.
We prove that the rules \textsc{sfl-comp} and \textsc{sfl-safe} are sound in Appendix~\ref{apx:action-composition-sound}.

\subsection{Hybrid Assume-Guarantee Contracts}
\label{sec:ag-hybrid}
In the prior section (Sec.~\ref{sec:ag-theory-sfl}), we introduced an action-based assume-guarantee theory.
However, this theory is not sufficient for verifying that specifications satisfy state-based properties.
For example, using action-based contracts alone, one cannot show that the Toy-2PC protocol satisfies the state invariant $\Box Consistent$ because action-based guarantees are action invariants.

Consequently, in this section, we introduce \textit{hybrid} contracts, in which the assumption is an action invariant and the guarantee is a state invariant.
We will define contracts and the local inductive invariant proof method in this theory, after which we will provide inference rules for composition.
We will show that, with the addition of hybrid contracts, our action-based theory can verify state-based properties of systems.

\paragraph{Hybrid contracts and local inductive invariants.}
We will use the notation $\sag{\alpha} C \sag{G}$ for hybrid contracts, where $G$ is a state-based guarantee that $C$ must satisfy under the action-based assumption $\alpha$.
We use the same notation for action-based and hybrid contracts for uniformity, since we will use both types of contracts for verifying a single system. 
The meaning of a hybrid contract is that the component $C$, when restricted to the assumption $\alpha$, must always satisfy the guarantee $G$.
Similarly to Sec.~\ref{sec:ag-theory-sfl}, we will restrict the component to the assumption algebraically, by taking the parallel composition of $C$ and $\mathcal{T}(\Box\alpha)$.
However, we require the guarantee to hold in a state-based fashion.
We now define hybrid contracts from this intuition.
\begin{definition}
    \label{def:hybrid-ag-contract}
    Let $C = (vars,Init,Next)^p$ be a PSTS.
    Also, let $\alpha$ be a non-temporal \flogic{} formula and $G$ be a non-temporal FOL formula such that such that $Params(\alpha) \cup Params(G) \subseteq p$, $\act \alpha \subseteq \act C$, and $SV(G) \subseteq vars$.
    A hybrid contract $\sag{\alpha} C \sag{G}$ is \textit{fulfilled} if and only if $\mathcal{T}(\Box\alpha) \parallel C \models \Box G$.
    We will denote contract fulfillment with the notation $\vdash \sag{\alpha} C \sag{G}$.
\end{definition}
For example, $\sag{\rho_1 \land \rho_2} ToyRM \sag{Consistent}$ is a hybrid contract.
Proving fulfillment requires the inductive invariant method; much like in the action-based theory, we translate the hybrid contract to a state-based contract in order to use the local inductive invariant method from Sec.~\ref{sec:contracts}.
For a hybrid contract $\sag{\alpha} C \sag{G}$, the corresponding state-based contract is $\ag{\mathcal{R}(\alpha)} \mathcal{B}(\Box\alpha) \parallel C \ag{G}$.
This translation is similar to the action-based contract translation, except the state-based formula $G$ does not need to be translated.
Based on this translation, we provide the following definition for showing that a contract is fulfilled by a local inductive invariant.
\begin{definition}
    \label{def:hybrid-local-ii}
    Let $\alpha$ be a non-temporal \flogic{} formula and let $G$ be a non-temporal FOL formula.
    We define $I \vdash \sag{\alpha} C \sag{G}$ to be fulfilled if and only if $I \vdash \ag{\mathcal{R}(\alpha)} \mathcal{B}(\Box\alpha) \parallel C \ag{G}$.
    Similarly to Sec.~\ref{sec:ag-theory-sfl}, we will prefer to write $\mathcal{B}(\phi)$ instead of $\mathcal{B}(\Box\phi)$ since the two are equivalent.
\end{definition}
The contract from above can be proved using the invariant $I_{RM}$ from Fig.~\ref{fig:toy-2pc-defs}.
In other words, we have $I_{RM} \vdash \sag{\rho_1 \land \rho_2} ToyRM \sag{Consistent}$.
The following theorem that shows that the local inductive invariant method is sufficient to show that a hybrid contract is fulfilled; we provide proof in Appendix~\ref{apx:hybrid-ii-sound}.
\begin{theorem}
    \label{thm:hybrid-ii-sound}
    $I \vdash \sag{\alpha} C \sag{G}$ implies $\vdash \sag{\alpha} C \sag{G}$
\end{theorem}

\paragraph{Composing contracts and local inductive invariants.}
Hybrid contracts are intended to be composed with action-based contracts.
Therefore, we present the inference rule \textsc{hybrid-comp} in Fig.~\ref{fig:hybrid-composition-rules} that composes an action-based contract with a hybrid contract, as well as their associated local inductive invariants.
Similarly to Sec.~\ref{sec:ag-theory-sfl}, we include $\mathcal{B}(\rho)$ in the conclusion of the rule \textsc{hybrid-comp} to ensure that the local invariant is well-defined.
The inference rule \textsc{hybrid-safe} in Fig.~\ref{fig:hybrid-composition-rules} additionally shows that \textsc{hybrid-comp} is sound, in the sense that it proves the composition of the components to be safe.
We prove that the two inference rules \textsc{hybrid-comp} and \textsc{hybrid-safe} are sound in Appendix~\ref{apx:hybrid-composition-sound}.
\begin{figure}
    \centering
    \begin{align*}
        &\inferrule*[lab=hybrid-comp]{I_1 \vdash \sag{\alpha} C_1 \sag{\rho} \\ I_2 \vdash \sag{\rho} C_2 \sag{G}}{I_1 \land I_2 \vdash \sag{\alpha} C_1 \parallel \mathcal{B}(\rho) \parallel C_2 \sag{G}}
        \quad
        \inferrule*[lab=hybrid-safe]{I_1 \vdash \sag{\alpha} C_1 \sag{\rho} \\ I_2 \vdash \sag{\rho} C_2 \sag{G}}{\vdash \sag{\alpha} C_1 \parallel C_2 \sag{G}}
    \end{align*}
    \caption{Assume-guarantee inference rules for transitive composition with an \flogic{} bridge formula $\rho$.}
    \label{fig:hybrid-composition-rules}
\end{figure}

We now return to the Toy-2PC example.
Recall from above $I_{RM} \vdash \sag{\rho_1 \land \rho_2} ToyRM \sag{Consistent}$, and also from Sec.~\ref{sec:ag-theory-sfl} that $I_{TM} \vdash \sag{\True} ToyTM \sag{\rho_1 \land \rho_2}$.
By the \textsc{hybrid-safe} rule, we can conclude that $\vdash \sag{\True} ToyTM \parallel ToyRM \sag{Consistent}$, i.e. that the entire system is safe.
Furthermore, the \textsc{hybrid-comp} rule shows that $I_{TM} \land I_{RM}$ is an inductive invariant \textit{for the entire system}.

\section{Compositional Inductive Invariant Inference}
\label{sec:ii-infer}
Using the two-layered assume-guarantee theory from the prior two sections, we now present a framework for compositional inductive invariant inference.
The problem statement is as follows: given a system $S$ and an invariant $\Box P$, find an inductive invariant that proves $S \models \Box P$.
Our compositional inference framework solves this problem with the following four steps: (1) system decomposition, (2) bridge formula inference, (3) local inductive invariant inference, and (4) global safety.
The visual in Fig.~\ref{fig:overview} shows the four steps at a high level; here, we provide details for each step, including examples using Toy-2PC.

\paragraph{(1) System decomposition.}
In this step, we decompose $S$ into components $C_1, \dots, C_n$ such that $S = C_1 \parallel \dots \parallel C_n$.
For example, in Toy-2PC, we decompose $Toy2PC$ into two components ($n=2$), $ToyRM$ and $ToyTM$ from Fig.~\ref{fig:toy-2pc-components}.

\paragraph{(2) Bridge formula inference.}
The goal of this step is to find formulas $\phi_1,\dots,\phi_{n+1}$ to create candidate contracts $\sag{\phi_i} C_i \sag{\phi_{i+1}}$.
The formulas $\phi_1,\dots,\phi_n$ are action invariants with the requirement that $\phi_1 = \True$.
The final formula $\phi_{n+1}$ is required to be the state-based safety property, i.e. $\phi_{n+1} = P$.
Consequently, the first $n-1$ contracts are action-based, while the final contract is hybrid.
In the Toy-2PC example, $\phi_1 = \True$, $\phi_2 = \rho_1 \land \rho_2$, and $\phi_3 = Consistent$, making the candidate contracts $\sag{\True} TM \sag{\rho_1 \land \rho_2}$ and $\sag{\rho_1 \land \rho_2} RM \sag{Consistent}$.

\paragraph{(3) Local inductive invariant inference.}
For each candidate contract $\sag{\phi_i} C_i \sag{\phi_{i+1}}$, the goal is to infer a local inductive invariant $I_i$ such that $I_i \vdash \sag{\phi_i} C_i \sag{\phi_{i+1}}$.
By Def.~\ref{def:sfl-local-ii}, the first $n-1$ contracts can be translated to the state-based contract $I_i \vdash \ag{\mathcal{R}(\phi_i)} \mathcal{B}(\phi_i) \parallel C_i \parallel \mathcal{B}(\phi_{i+1}) \ag{\mathcal{R}(\phi_{i+1})}$.
By Def.~\ref{def:hybrid-local-ii}, we can also translate the final contract to the hybrid contract $I_i \vdash \ag{\mathcal{R}(\phi_i)} \mathcal{B}(\phi_i) \parallel C_i \ag{\phi_{i+1}}$.
In both cases, we obtain a state-based contract that entails a local inductive invariant inference problem for finding $I_i$.

We describe how to solve the local inductive invariant inference problem for the first $n-1$ contracts; the solution for the final hybrid contract is similar.
Suppose that $\mathcal{B}(\phi_i) \parallel C_i \parallel \mathcal{B}(\phi_{i+1}) = (vars, Init, Next)^p$, then we construct $D_i = (vars, Init \land \mathcal{R}(\phi_i), Next \land \mathcal{R}(\phi_i)')^p$.
By Def.~\ref{def:ag-contract}, any inductive invariant that proves $D_i \models \Box \mathcal{R}(\phi_{i+1})$ also serves as a local inductive invariant $I_i$ for the corresponding contract.
However, verifying $D_i \models \Box \mathcal{R}(\phi_{i+1})$ reduces to a global inductive invariant inference problem, which can be solved using off-the-shelf tools.

For Toy-2PC, the state-based proof obligations are to find $I_1$ and $I_2$ such that $I_1 \vdash \ag{\True} ToyTM \parallel ToyB \ag{ToyR}$ and $I_2 \vdash \ag{ToyR} ToyB \parallel ToyRM \ag{Consistent}$.
The proof obligation for the first contract reduces to finding an inductive invariant $I_1$ for $ToyTM \parallel ToyB \models \Box ToyR$, while the proof obligation for the second contract reduces to finding an inductive invariant $I_2$ for $\mathcal{T}(\Box(\rho_1 \land \rho_2)) \parallel ToyRM \models \Box Consistent$.
In this case, $I_1 = I_{TM}$ and $I_2 = I_{RM}$ are sufficient choices.

We also point out that local inductive invariant inference can be parallelized because each contract constitutes an independent proof obligation.
We will utilize this fact during our evaluation in Sec.~\ref{sec:evaluation}.

\paragraph{(4) Global safety.}
Let $I$ be the conjunction of each local inductive invariant, i.e. $I = I_1 \land \dots \land I_n$.
Also, let $B$ be the parallel composition of all bridge components, i.e. $B = \mathcal{B}(\phi_1) \parallel \dots \parallel \mathcal{B}(\phi_n)$.
Then, by the rules \textsc{sfl-comp} and \textsc{hybrid-comp}, we can infer that $I \vdash \sag{\True} S \parallel B \sag{P}$.
Furthermore, by the rules \textsc{sfl-safe} and \textsc{hybrid-safe}, we can infer that $\vdash \sag{\True} S \sag{P}$.
In other words, the entire system is safe and $I$ is an inductive invariant for the entirety of $S$.
In the Toy-2PC example, $I_{TM} \land I_{RM}$ is an inductive invariant $I$ for the entire protocol.


\section{Evaluation}
\label{sec:evaluation}
We  evaluate our inductive invariant inference method by applying it to two case studies of distributed protocols, both of which are specified in \TLA{}.
We compare the compositional approach to the global approach using Endive \cite{schultz:2022indinvs} as a baseline.
Endive is a state-of-the-art tool for automatic inductive invariant inference for \TLA{} specifications.
In addition, we use Endive as our off-the-shelf tool for local inductive invariant inference, as described in our compositional framework (Sec.~\ref{sec:ii-infer}).

The first case study is the full version of the Two Phase Commit protocol \cite{Gray:2006}, in which we compare automatic inductive invariant inference using both the global and compositional methods.
The second case study is an industrial-scale protocol of a Raft \cite{Ongaro:2014} style algorithm that runs at MongoDB \cite{Mongodb}.
To the best of our knowledge, this case study is beyond the reach of any existing algorithm for automatic inductive invariant inference, including Endive.
However, with the compositional approach, we infer an inductive invariant for the protocol semi-automatically.
Between the two case studies, we provide evidence that, in comparison to the global approach, the compositional approach can be more efficient for inductive invariant inference and results in more concise proofs.
Additionally, we show that the artifacts from compositional inference, including bridge formulas and local inductive invariants, provide insights into the behavior of the protocols beyond those of the global approach.

We manually decomposed all specifications in our case studies; automatic specification decomposition is well-studied \cite{Cobleigh:2006,Dardik:2024,Metzler:2008,Nam:2008} and beyond the scope of this paper.
We also manually created all bridge formulas in the case studies.
Bridge formula inference is a nontrivial task whose automation is a significant research problem on its own and is beyond the scope of this paper.
We also include a machine checked proof for each inductive invariant written in the \TLA{} Proof System (TLAPS) \cite{cousineau:2012tla} proof language.
All experiments and proofs were run on MacOS version 14.7.6 with an M1 Pro chip, and are available in our supplementary materials.

\subsection{Case Study 1: Two Phase Commit}
\label{sec:2pc-case-study}
In this case study, we verify the (full) Two Phase Commit (2PC) protocol \cite{Gray:2006}.
We compare automatic inductive invariant inference between our compositional framework and the global approach.
\begin{figure}
    \begin{align*}
        \rho_r \triangleq &\land (\exists rm \in RMs : onceRcvCommit(rm)) \implies (\forall rm \in RMs : onceSndPrepare(rm)) \\
        &\land (\exists rm \in RMs : onceRcvAbort(rm)) \implies (\forall rm \in RMs : \neg onceRcvCommit(rm))
    \end{align*}
    \begin{align*}
        \rho_e \triangleq &\land (\exists rm \in RMs : onceSndCommit(rm)) \implies (\forall rm \in RMs : onceRcvPrepare(rm)) \\
        &\land (\exists rm \in RMs : onceSndAbort(rm)) \implies (\forall rm \in RMs : \neg onceSndCommit(rm))
    \end{align*}
    \caption{Definitions for the bridge formulas in the 2PC protocol.}
    \label{fig:2pc-full-bridge}
\end{figure}
\begin{figure}
    \begin{align*}
        J_g \triangleq\ \ &\forall rm \in RM : (tmState = ``init" \land [type \mapsto ``Prepared", theRM \mapsto rm] \in msgs) \implies \\
        &\hspace{3mm}(rmState[rm] = ``prepared") \\
        J_l \triangleq\ \ &\forall rm \in RMs : onceSndPrepare[rm] \implies rmState[rm] \neq ``working"
    \end{align*}
    \caption{$J_g$ is a conjunct of the global inductive invariant for 2PC while $J_l$ is a conjunct of a local inductive invariant for 2PC.
        Both formulas specify that resource managers cannot abort after sending a prepare message.}
    \label{fig:2pc-full-conjs}
\end{figure}

\paragraph{Protocol description.}
The 2PC protocol allows messages to be delayed, dropped, and reordered, while Toy-2PC assumes a perfect network.
Rather than single $Prepared$, $Commit$ and $Abort$ actions, this protocol models sending and receiving each message over a network.
Therefore, 2PC includes the actions $SndPrepared$ and $RcvPrepared$, etc.
The protocol also includes an additional state variable $msgs$ that records all messages that have been sent across the network; message delay, dropping, and reordering are modeled as a receiver ignoring a message for some period of time (e.g., indefinitely for a dropped message).

\paragraph{Inductive invariant inference.}
For global inference, we use the inductive invariant and the proof reported in the Endive paper \cite{schultz:2022indinvs}.
For compositional inference, we use the framework presented in Sec.~\ref{sec:ii-infer}.
We outline the four steps of the framework below.

Step (1) \textit{system decomposition.}
We decompose the system into three components $RM$, $Env$, and $TM$.
The components respectively represent the resource managers, the environment, and the transaction manager.

Step (2) \textit{bridge formula inference.}
The candidate contracts are $\sag{\True} TM \sag{\rho_e}$, $\sag{\rho_e} Env \sag{\rho_r}$, and $\sag{\rho_r} RM \sag{Consistent}$, where the bridge formulas $\rho_r$ and $\rho_e$ are shown in Fig.~\ref{fig:2pc-full-bridge}.
The fluents in $\rho_r$ and $\rho_e$ follow the \textit{once} pattern, e.g., from Fig.~\ref{fig:toy-2pc-fluents}; we therefore omit the explicit definition of the fluents for brevity.
Intuitively, $\rho_r$ specifies the acceptable behavior of the network, while $\rho_e$ specifies the acceptable behavior of the transaction manager.
The network requirement $\rho_r$ is conceptually the same as $\rho_1 \land \rho_2$ (Fig.~\ref{fig:toy-2pc-sfl-formula}) from the Toy-2PC protocol, except $\rho_r$ describes the requirement in terms of messages that the resource managers can receive from the network.
Likewise, the transaction manager requirement $\rho_e$ is conceptually the same as $\rho_1 \land \rho_2$, except $\rho_e$ describes messages that the transaction manager may send over the network.

Step (3) \textit{local inductive invariant inference.}
We first translate each of the three contracts to their state-based counterpart.
After, we use Endive to infer a local inductive invariant.
We include a TLAPS proof for each inductive invariant that Endive returned.

Step (4) \textit{global safety.}
By proof rules \textsc{sfl-comp} and \textsc{hybrid-comp}, the conjunction of the three local inductive invariants is an inductive invariant for the entire 2PC protocol.
Furthermore, we infer that the protocol is safe by rules \textsc{sfl-safe} and \textsc{hybrid-safe}.

\paragraph{Results.}
\begin{table}[]
    \centering
    \caption{Comparison between global and compositional inductive invariant inference for 2PC.
        All times from Endive are reported in seconds.}
    \label{tab:2pc-compare}
    \begin{tabular}{lrrr}
         \textbf{Specification} & \textbf{Endive Time} & \textbf{Proof Size} & \textbf{Inv. Size} \\ 
         \hline
         $TwoPhase$ (global) & 34.82 & 107 & 10 \\
         \hline
         $RM$ & 15.44 & 8 & 5 \\
         $Env$ & 30.43 & 86 & 8 \\
         $TM$ & 6.13 & 8 & 5 \\
    \end{tabular}
\end{table}
We summarize our results in Table.~\ref{tab:2pc-compare}. 
For each specification, we include the run-time for Endive in seconds (Endive Time), the size of the corresponding proof (Proof Size), and the number of conjuncts in the inductive invariant (Inv. Size) as a rough measure for its size.
When reporting the size of the proof, we count only the number of proof steps to avoid counting irrelevant lines that result due to formatting.

\paragraph{Discussion.}
\textbf{Efficiency of invariant inference.}
In terms of run-time for automatic inductive invariant inference, the global approach is the least efficient at 34.82 seconds.
By parallelizing local inductive invariant inference, as recommended in Sec.~\ref{sec:ii-infer}, compositional inference is performed more efficiently (30.43 seconds) than the global approach.
This comparison provides evidence that compositional inference can be more efficient than the global approach.

\textbf{Proof conciseness.}
Table.~\ref{tab:2pc-compare} shows that the proof sizes for the local inductive invariants, both individually and combined, are smaller than the proof for the global invariant.
These results provide evidence that the proofs for local inductive invariants may be more concise than the proofs for global invariants.
We also point out that the number of conjuncts in each local inductive invariant is smaller than the number in the global invariant; however, we consider the number of conjuncts to be a rough, secondary measure for evaluating the conciseness of the proof for an inductive invariant.

\textbf{Insights from compositional inference.}
In addition to their value for verification, inductive invariants are also valuable for the insights they provide about a system.
The artifacts from compositional inference--bridge formulas and local inductive invariants--also provide insights about systems, some of which are not found in global inductive invariants.
For example, the two conjuncts of both bridge formulas $\rho_r$ and $\rho_e$ offer a concise modular specification of the two phases of the 2PC protocol:
the first conjuncts specify that all resource managers attempt to prepare in the first phase before a commit, while the second conjuncts specify that the transaction manager chooses to either commit or abort in the second phase.
However, no single part of the global inductive invariant offers this concise insight.

\textbf{Local invariants offer concise local insights.}
Furthermore, we observe that local insights tend to be encoded more concisely in local invariants in comparison with global invariants.
For example, a key requirement for the 2PC protocol is that once each resource manager sends a prepare message, it can no longer silently abort.
The global inductive invariant and the local inductive invariant for $RM$ both encode this requirement differently; we show the relevant conjunct from each invariant in Fig.~\ref{fig:2pc-full-conjs}.
Notice that the global invariant refers to the state variables from all three components of the system ($RM$, $Env$, and $TM$), while the local invariant refers only to the state variables and actions for the relevant component ($RM$).
By referring to only the relevant component, the local inductive invariant reduces the scope of the protocol that is needed to understand the key requirement.

\subsection{Case Study 2: Mongo Static Raft}
\label{sec:mongo-case-study}
In this case study, we verify Mongo Static Raft (MSR) \cite{Schultz:2021}, an industrial-scale protocol based on Raft \cite{Ongaro:2014} that runs at MongoDB \cite{Mongodb}.
We describe and compare our effort to infer an inductive invariant for MSR using both the global and compositional approaches.

\paragraph{Protocol description.}
MSR is a distributed consensus protocol that uses leader elections to maintain a consistent replicated state machine.
The protocol begins with an election, where a server within a cluster is decided to be the \textit{primary}.
Each server updates its metadata with the primary server and the \textit{term}--a natural number that increases with each election--for which the primary will serve.
The primary server accepts client requests and tracks them in a local data structure called a \textit{log}.
The servers in the cluster replicate the contents of their log in gossip style to the other nodes in the cluster.
Once a majority of servers add a client request to their log, the request is considered to be \textit{committed} and each server adds it to their state machine.

The safety property for the protocol is state machine safety (SMS), which we show formally in Fig.~\ref{fig:mongo-sms}.
This property specifies that, for any two commits $c_1$ and $c_2$, if they share the same index ($c_1.ind = c_2.ind$), then the two commits must be identical ($c_1 = c_2$).
In general, a commit $c$ records the index in the log ($c.ind$) and the term in which the client request occurred ($c.term$).
The MSR specification uses the term ($c.term$) as an abstract representation for the contents of a client request.

In the MSR protocol, the network can drop or reorder messages and can also cause the servers to be partitioned.
However, servers are expected to act in a failstop manner, meaning that failed servers do not produce faulty or corrupt messages.
When a leader fails or if the network is partitioned, the protocol will elect a new primary in a new term.
All servers have the ability to ``catch up'' other servers in a lower term by rolling back extraneous log entries and adding correct ones.
Additionally, the protocol derives its name from the assumption that its configuration--the set of nodes in the cluster--is static.

In the \TLA{} encoding of the protocol \cite{Schultz:2021}, the set of servers is given as a parameter.
The specification includes the following state variables: $committed$, $log$, $state$, and $currentTerm$, which represent the set of commits, the log, whether each server is a primary, and the current election term.
The variable $committed$ keeps track of the set of all commits throughout the protocol, while the remaining three variables represent data structures that are local to each server.

\begin{figure}
    \begin{subfigure}[b]{0.47\textwidth}
        \vspace{2mm}
        {\small \input{tla/msr_commitAtTermInd}}
        \caption{Fluent definition for \textit{commitAtTermInd}.}
        \label{fig:mongo-commmitAtTermInd}
    \end{subfigure}
    \hfill
    \begin{subfigure}[b]{0.48\textwidth}
        \vspace{2mm}
        {\small \input{tla/msr_leaderAtTermServ}}
        \caption{Fluent definition for \textit{leaderAtTermServ}.}
        \label{fig:mongo-leaderAtTermServ}
    \end{subfigure}
    \begin{subfigure}[b]{0.99\textwidth}
        \vspace{2mm}
        {\small \input{tla/msr_currentLeader}}
        \caption{Fluent definition for \textit{currentLeader}.}
        \label{fig:mongo-currentLeader}
    \end{subfigure}
    \begin{subfigure}[t]{0.99\textwidth}
        \begin{align*}
            &StateMachineSafety \triangleq \forall c_1,c_2 \in committed : (c_1.ind = c_2.ind) \implies c_1 = c_2
        \end{align*}
        \caption{The key safety property $StateMachineSafety$, or $SMS$ for short.}
        \label{fig:mongo-sms}
    \end{subfigure}
    \begin{subfigure}[t]{0.99\textwidth}
        \begin{align*}
            \rho_c \triangleq\ &\forall i \in \mathbb{N} : \exists t_1 \in \mathbb{N} : \forall t_2 \in \mathbb{N} : commitAtTermInd(t_2,i) \implies (t_1 = t_2) \\
            \rho_l^* \triangleq\ &\forall t \in \mathbb{N} : \forall s_1,s_2 \in Server : \\
            &\hspace{3mm}(leaderAtTermServ(t,s_1) \land leaderAtTermServ(t,s_2)) \implies (s_1 = s_2)
        \end{align*}
        \caption{$\rho_c$ is the bridge formula for the $Committed$ and $Log$ components.
            $\rho_l^*$ is one conjunct of $\rho_l$, the bridge formula for the $Log$ and $StateTerm$ components.}
        \label{fig:mongo-bridges}
    \end{subfigure}
    \caption{Definitions for the safety property and bridge formulas in the Mongo Static Raft protocol.}
\end{figure}

\paragraph{Inductive invariant inference.}
We attempted to use Endive to automatically infer a global inductive invariant for MSR.
However, Endive exits with failure for global inference because it could not find predicates to eliminate all counter examples to induction (CTIs).
Therefore, we manually constructed a global inductive invariant for MSR, which we proved with TLAPS.

With our compositional framework, we inferred an inductive invariant for MSR semi-automatically.
We also used Endive in attempt to automatically infer each local inductive invariant.
Endive completed successfully for one component and failed for two components, where the failures were also because the tool could not eliminate all CTIs.
We now describe our compositional inference effort in detail for the four steps from Sec.~\ref{sec:ii-infer}.

Step (1) \textit{system decomposition.}
We decompose MSR into three components $Committed$, $Log$, and $StateTerm$.
The components respectively represent the set of commits, the log for each server, and the state and term metadata for each server.

Step (2) \textit{bridge formula inference.}
The candidate contracts are $\sag{\True} StateLog \sag{\rho_l}$, $\sag{\rho_l} Log \sag{\rho_c}$, and $\sag{\rho_c} Committed \sag{SMS}$.
In Fig.~\ref{fig:mongo-bridges}, we show the bridge formula $\rho_c$ as well as one exemplar conjunct $\rho_l^*$ for the bridge formula $\rho_l$.
We show the entire bridge formula $\rho_l$ in Appendix~\ref{apx:mongo-fluent-defs}, Fig.~\ref{fig:mongo-apx}.
We show the fluent definitions (transition systems only) for \textit{commitAtTermInd}, \textit{leaderAtTermServ}, and \textit{currentLeader} in Figures \ref{fig:mongo-commmitAtTermInd}, \ref{fig:mongo-leaderAtTermServ}, and \ref{fig:mongo-currentLeader} respectively.
The definitions for the remaining fluents are included in Appendix~\ref{apx:mongo-fluent-defs}.

The fluents \textit{commitAtTermInd} and \textit{leaderAtTermServ} use the \textit{once} style as seen in the 2PC and Toy-2PC protocols.
These three fluents respectively represent whether a commit has happened at a specific term and index, and whether a server has been elected leader at a specific term.
The fluent \textit{currentLeader}, however, represents servers who believe themselves to be a leader at a specific term; this may be more than one server in the case of a network partition.
In the definition for this fluent (Fig.~\ref{fig:mongo-currentLeader}), a server believes itself to be a leader upon a \textit{BecomeLeader} action, but no longer once its log is updated via a \textit{GetEntries} or \textit{RollbackEntries} action.


Step (3) \textit{local inductive invariant inference.}
For each contract, we attempted to automatically infer a local inductive invariant with Endive.
Endive failed for the $Log$ and $StateTerm$ components, but succeeded for the $Committed$ component.
For the two cases that Endive failed, we manually constructed a local inductive invariant.
We also include a TLAPS proof for all local inductive invariants.

Step (4) \textit{global safety.}
By proof rules \textsc{sfl-comp} and \textsc{hybrid-comp}, the conjunction of the three local inductive invariants is an inductive invariant for the entire MSR protocol.
Furthermore, we infer that the protocol is safe by rules \textsc{sfl-safe} and \textsc{hybrid-safe}.

\begin{table}[]
    \centering
    \caption{Comparison between global and compositional inductive invariant inference for Mongo Static Raft.}
    \label{tab:mongo-compare}
    \begin{tabular}{lrrr}
         \textbf{Specification} & \textbf{Endive Time} & \textbf{Proof Size} & \textbf{Inv. Size} \\ 
         \hline
         $MSR$ (global) & NA & 1,686 & 17 \\
         \hline
         $Committed$ & 19.80 & 61 & 2 \\
         $Log$ & NA & 123 & 9 \\
         $StateTerm$ & NA & 188 & 7
    \end{tabular}
\end{table}
\paragraph{Results.}
We summarize our results for each specification in Table.~\ref{tab:mongo-compare}.
We include the run-time for Endive in seconds (Endive Time), unless the tool failed in which case we write NA.
We also report the size of the proof for the inductive invariant (Proof Size) and the number of conjuncts in the invariant (Inv. Size). 
Similarly to the 2PC case study (Sec.~\ref{sec:2pc-case-study}), we count only the number of proof steps when we report the proof size.

\paragraph{Discussion.}
\textbf{Efficiency of invariant inference.}
Table.~\ref{tab:mongo-compare} shows that Endive fails for the global inductive invariant inference task.
In the compositional approach, Endive fails for the \textit{Log} and \textit{StateTerm} components so we manually created local inductive invariants for these two components.
However, Endive completed successfully for the \textit{Committed} component, which is smaller and less complex than the other two components.
Ultimately, we inferred an inductive invariant semi-automatically with the compositional approach, which provides evidence that compositional inductive invariant inference can be more efficient than global inference.
This result also suggests that the compositional approach may benefit from further research into system decomposition, specifically for finding simpler components.

\textbf{Proof conciseness.}
Table.~\ref{tab:mongo-compare} shows that the length of the TLAPS proof for each individual component is shorter than the proof for the global invariant by an order of magnitude.
This result provides evidence that the proof for local inductive invariants may be shorter than the proofs for global inductive invariants.
We also point out that the local inductive invariants have fewer conjuncts, but offer this observation as a rough, secondary metric.

\textbf{Insights from compositional inference.}
Much like the 2PC case study, we found that the artifacts from compositional inference provided valuable insights about the case study that are not captured by the global inductive invariant.
For example, consider $\rho_l^*$ in Fig.~\ref{fig:mongo-bridges}, which is a conjunct of the bridge formula $\rho_l$.
The formula $\rho_l^*$ represents the \textit{one leader per term} property, which is a key lemma at the heart of the correctness argument for Raft \cite{Ongaro:2014}.
While this property does appear in the global inductive invariant, the fact that it appears in the bridge formula $\rho_l$ offers the more specific insight that \textit{one leader per term} is guaranteed by the $StateTerm$ component and assumed by the $Log$ component.
This insight could provide the basis for modular optimizations to the protocol, e.g., using a new leader election scheme in the $StateTerm$ component that maintains the property.

The presence of important properties such as \textit{one leader per term} in bridge formulas may also be a reason that the local inductive invariant proofs are more concise than the global proof.
Using the \textit{Log} component as an example, the intuition for this reasoning is that local inductive invariant inference becomes simpler when it can assume key lemmas such as \textit{one leader per term}, rather than needing to infer them.

\textbf{Simpler constraints in local invariants.}
We observe that local inductive invariants can specify simpler constraints when compared with the global invariant.
For example, consider the conjunct $K_g$ from the global invariant in Fig.~\ref{fig:mongo-in-log-conjs}.
This conjunct specifies an intricate condition for when an entry must appear in the logs of a secondary (non-primary) server.
The purpose of this conjunct is to help establish leader completeness properties in the induction proof.
However, the compositional approach avoids tying the log to leader completeness properties, which results in a simpler constraint than $K_g$.
The fact that the leader completeness properties do not depend on the log is explicit in the compositional approach because the $Log$ component assumes all its leader completeness properties from the bridge formula $\rho_l$.
As a result, the local inductive invariant for \textit{Log} specifies simpler constraints on the log.
We show the lone conjunct $K_l$ that constrains the log in Fig.~\ref{fig:mongo-in-log-conjs}.


\begin{figure}
    \begin{align*}
        K_g \triangleq\ &\forall s \in Server : (state[s] = Secondary \land LastTerm(log[s]) = currentTerm[s]) \implies \\
        &\hspace{3mm}\lor\ \exists p \in Server : \\
        &\hspace{6mm}\land\ state[p] = Primary \\
        &\hspace{6mm}\land\ currentTerm[p] = currentTerm[s] \\
        &\hspace{6mm}\land\ LastTerm(log[p]) \geq LastTerm(log[s]) \\
        &\hspace{6mm}\land\ Len(log[p]) \geq Len(log[s]) \\
        &\hspace{3mm}\lor\ \exists p \in Server : \\
        &\hspace{6mm}\land\ state[p] = Primary \\
        &\hspace{6mm}\land\ currentTerm[p] > currentTerm[s] \\
        &\hspace{3mm}\lor\ \forall t \in Server : state[t] = Secondary \\
        K_l \triangleq\ &\forall t_1,t_2,i \in \mathbb{N} : \forall s \in Server : \\
        &\hspace{3mm}(commitAtTermInd[t_1][i] \land currentLeader[s][t_2] \land t_1 \leq t_2) \implies log[s][i] = t_1 
    \end{align*}
    \caption{$K_g$ is a conjunct from the global invariant and $K_l$ is a conjunct from the local invariant for \textit{Log}.
        Both formulas specify when an entry must be present in the log.}
    \label{fig:mongo-in-log-conjs}
\end{figure}

\section{Related Work}
\label{sec:related-work}
Many techniques exist for automatically inferring inductive invariants.
Approaches include interpolation \cite{mcmillan:2003}, ICE learning \cite{Garg:2014}, syntax-driven enumeration of invariants \cite{Hance:2021,Yao:2021,Yao:2022}, and incremental lemma learning \cite{bradley:2011,Goel:2021,Koenig:2020,Koenig:2022,schultz:2022indinvs}.
Theoretical work based on the \textit{Hoare query model} has shown that the incremental approach--specifically with relative induction checks--can learn inductive invariants exponentially faster than with induction checks alone \cite{Feldman:2019}.
A more recent approach, based on the observation that counter examples to induction (CTIs) help inductive invariant inference algorithms to find new predicates and vice versa (duality), attempts to balance the search for CTIs and predicates to achieve progress theorems \cite{Padon:2022}.
While the above techniques have certain advantages, each one experiences scalability issues by attempting to infer an inductive invariant for the transition relation of an entire system.
In contrast, our compositional framework divides the transition relation to improve the efficiency of inductive invariant inference.
However, the techniques above are complementary to our compositional framework because they can be used for local inductive invariant inference.

Due to the limitations of fully automated inductive invariant inference, frameworks have been proposed for assisting users in finding an inductive invariant.
Ivy \cite{Padon:2016} and similar languages \cite{Wilcox:2024} are popular frameworks that encourage a user to write specifications in a decidable fragment of FOL.
Kondo \cite{Zhang:2024kondo} is also a tool that requires user assistance, but is specifically designed to strike a balance between automation and manual proof effort.
Lamport describes a technique for finding inductive invariants for \TLA{} \cite{Lamport2018UsingTT} in which a person manually finds CTIs using the TLC model checker \cite{Yu:1999}.
Each of these techniques also runs into scalability issues, both from to the manual efforts involved and the fact that a user attempts to find a global inductive invariant for the entire transition relation.

There is a large body of work for automating compositional verification for finite-state processes.
Beginning with the seminal paper of Cobleigh et al. \cite{Cobleigh:2003}, researchers have attempted to find ways to learn and compute assumptions (which we call bridges), both explicitly \cite{alur:2005,Chen:2009,Cobleigh:2003,Gupta:2007,Nam:2008} and implicitly (symbolically) \cite{Chen:2010,He:2014}.
While these techniques are powerful, they do not, in general, apply to parameterized systems because they are inherently finite-state.
Our framework allows assume-guarantee contracts to be proved with local inductive invariants, which allows us to compositionally verify parameterized systems.

The method of Owicki Gries \cite{Owicki:1976}, Rely-Guarantee Reasoning \cite{Jones:1983}, and Concurrent Separation Logic \cite{OHEARN:2007} are each assume-guarantee instances for proving properties of programs; however, we are interested in verifying declarative specifications in this work.


Past Time Temporal Logic (PTL) \cite{Kamp:1968} is a variant of temporal logic with operators that can specify propositions that happened in the past.
For example, the \textit{once} operator specifies that a proposition was true at least once in the past and can be seen as the past-time version of the \textit{eventually} ($\lozenge$) temporal operator.
The \textit{once} operator is analogous to the \textit{once} style fluents from the 2PC and Toy-2PC protocols.
However, the expressiveness of PTL (and LTL) with actions is limited, which inspired the Fluent LTL specification language \cite{Giannakopoulou:2003}.
In our work, we introduce \flogic{} that builds upon Fluent LTL by including quantifiers and symbolic fluents that accept arguments.
Ultimately, the extensions in \flogic{} make the language appropriate for specifying parameterized systems.

\section{Limitations and Future Work}
\label{sec:future-work}
In this paper, we presented a compositional inductive invariant inference framework that is based on a two layer assume-guarantee theory.
As part of our presentation, we showed that the inference rules for contract composition are sound.
However, we leave a formal treatment of completeness for future work.
We also plan to investigate whether the inference rules have the cut elimination property, which may reveal a technique for eliminating auxiliary variables from inductive invariants that are inferred compositionally.

We created all bridge formulas in this paper manually because automated inference of bridge formulas is a nontrivial task.
In the future, we plan to create an algorithm for automatic inference of bridge formulas, specifically for the \flogic{} language.
We also plan to investigate whether \flogic{} is useful for tasks besides compositional invariant inference.
For example, we plan to explore whether we can compute the weakest assumption \cite{Giann:2002} of a system in \flogic{}, which may be useful for robustness analysis \cite{Zhang:2020}.

In our evaluation (Sec.~\ref{sec:evaluation}), we were able to use compositional inductive invariant inference to verify the Mongo Static Raft protocol semi-automatically.
We automatically inferred a local invariant for the smallest component, which suggests that our technique would benefit from more granular decompositions.
In the future, we plan to investigate techniques for both automated and more granular decompositions.
Ultimately, we plan to fully automate our compositional verification framework in Sec.~\ref{sec:ii-infer}, in which decomposition, bridge formula inference, and local invariant inference are all automated.

\begin{acks}
This project was supported by a \TLA{} Foundation grant.
\end{acks}

%
%

\printbibliography

\pagebreak
\appendix
\section*{Appendix}

\section{Proof of Thm.~\ref{thm:ii-implies-safe}}
\label{apx:ii-implies-safe}
We now prove Thm.~\ref{thm:ii-implies-safe} which states: $I \vdash \ag{A} C \ag{G}$ implies $\vdash \ag{A} C \ag{G}$.
\begin{proof}
    Assume that $I \vdash \ag{A} C \ag{G}$.
    Using the inductive invariant proof method, the three equations in Def.~\ref{def:local-ind-inv} establish $(vars, Init \land A, Next \land A \land A')^p \models \Box G$.
    However, notice that $(vars, Init \land A, Next \land A \land A')^p$ is identical to $(vars, Init \land A, Next \land A')^p$, from which we can infer $\vdash \ag{A} C \ag{G}$ by Def.~\ref{def:ag-contract}.
\end{proof}

\section{Proof of Soundness for State-Based Contract Composition}
\label{sec:theory-soundness}
\begin{figure}
    \centering
    \begin{align*}
        &\inferrule*[lab=interf-free]{I \vdash \ag{A} C_1 \ag{G}}{I \vdash \ag{A} C_1 \parallel C_2 \ag{G}}
        \quad
        \inferrule*[lab=trans-inv]{I_1 \vdash \ag{A} C \ag{R} \\ I_2 \vdash \ag{R} C \ag{G}}{I_1 \land I_2 \vdash \ag{A} C \ag{G}}
    \end{align*}
    \caption{Basic inference rules for proving soundness of the composition rules.}
    \label{fig:soundness-rules}
\end{figure}
In this appendix, we show that the composition rules for the state-based theory are sound.
Our proofs are based on two observations, which we encode formally as inference rules.
The first observation is that components do not share state variables, and hence are interference-free.
This observation is captured by inference rule \textsc{interf-free}, which we present in Fig.~\ref{fig:soundness-rules}.
The second observation, captured by inference rule \textsc{trans-inv} in Fig.~\ref{fig:soundness-rules}, is that each individual component has the \textit{transitivity of invariance} property.
We will now show that these two inference rules are sound, which we will then use as lemmas to prove that the composition rules are sound.
\begin{lemma}
    \label{lem:interf-free}
    \textsc{interf-free} is sound.
\end{lemma}
\begin{proof}
    Suppose that $I \vdash \ag{A} C_1 \ag{G}$, then by definition we have the following three facts: (1) $Init_1 \land A \implies I$, (2) $I \land Next_1 \land A \land A' \implies I'$, and (3) $I \implies G$.
    Our goal is to prove that $I \vdash \ag{A} C_1 \parallel C_2 \ag{G}$.
    Because $Init_1 \land Init_2 \land A \implies Init_1 \land A$ and also because of fact (1), we can conclude initiation.
    Also, due to fact (3), $I$ implies safety.
    Therefore, it remains to prove consecution.
    
    We will prove consecution by the possible cases on the actions of $C_1 \parallel C_2$.
    The three cases for an action $a \in \act (C_1 \parallel C_2)$ are (i) $a \in (\act C_1) \cap (\act C_2)$, (ii) $a \in \act C_1$ and $a \notin \act C_2$, or (iii) $a \notin \act C_1$ and $a \in \act C_2$.
    In each possible case, will prove consecution, i.e. that $I \land a \land A \land A' \implies I'$.
    \begin{enumerate}[label=(\roman*)]
        \item In this case, $a \implies Next_1$.
            Consecution follows due to this as well as fact (2).
        \item We have $a \implies Next_1$ in this case as well, and therefore consecution also follows due to fact (2).
        \item In this case, $a \implies SV(C_1)' = SV(C_1)$, where $SV(C_1)' = SV(C_1)$ is an abuse of notation that indicates that the state variables of $C_1$ are unchanged.
            Furthermore, by Def.~\ref{def:local-ind-inv}, we have $SV(I) \subseteq SV(C_1)$ which implies $a \implies SV(I)' = SV(I)$.
            Finally, consecution follows because $I \land SV(I)' = SV(I) \implies I'$.
    \end{enumerate}
    Therefore, we have shown $I \vdash \ag{A} C_1 \parallel C_2 \ag{G}$ as desired.
\end{proof}
\begin{lemma}
    \label{lem:trans-inv}
    \textsc{trans-inv} is sound.
\end{lemma}
\begin{proof}
    Suppose that $I_1 \vdash \ag{A} C \ag{R}$ and $I_2 \vdash \ag{R} C \ag{G}$.
    Then we have the following three facts: (1) $Init \implies I_1$ and $Init \implies I_2$, (2) $I_1 \land Next \land A \land A' \implies I_1'$ and $I_2 \land Next \land R \land R' \implies I_2'$, and (3) $I_1 \implies R$ and $I_2 \implies G$.
    Our goal is to prove that $I_1 \land I_2 \vdash \ag{A} C \ag{G}$.
    Initiation and safety follow by facts (1) and (3) respectively.
    Then, the following equations establish consecution, where the right-hand side column indicates the reason for each implication in parentheses.
    \begin{align}
        &I_1 \land I_2 \land Next \land A \land A' && \\
        \implies & I_1 \land R \land I_2 \land Next \land A \land A' && (I_1 \implies R) \label{eqn:trans-inv1} \\
        \implies & I_1' \land R \land I_2 \land Next && (I_1 \land Next \land A \land A' \implies I_1') \label{eqn:trans-inv2} \\
        \implies & I_1' \land R \land R' \land I_2 \land Next && (I_1' \implies R') \label{eqn:trans-inv3} \\
        \implies & I_1' \land I_2' && (I_2 \land Next \land R \land R' \implies I_2') \label{eqn:trans-inv4}
    \end{align}
    In the equations above, the reasons in (\ref{eqn:trans-inv1}) and (\ref{eqn:trans-inv3}) follow due to fact (3), while the reasons in (\ref{eqn:trans-inv2}) and (\ref{eqn:trans-inv4}) follow due to fact (2).
\end{proof}
The two lemmas above show that the \textsc{interf-free} and \textsc{trans-inv} inference rules (Fig.~\ref{fig:soundness-rules}) are sound.
We will now use these two lemmas to show that the compositional inference rules (Fig.~\ref{fig:composition-rules}) are sound.
We also include a proof that \textsc{naive-comp} is sound for completness.
\begin{theorem}
    \textsc{naive-comp} is sound.
\end{theorem}
\begin{proof}
    Suppose that $I_1 \vdash \ag{A} C_1 \ag{R}$ and $I_2 \vdash \ag{R} C_2 \ag{G}$.
    By Lemma~\ref{lem:interf-free} and the fact that $\parallel$ is commutative, we have $I_1 \vdash \ag{A} C_1 \parallel C_2 \ag{R}$ and $I_2 \vdash \ag{R} C_1 \parallel C_2 \ag{G}$.
    Finally, the theorem follows with Lemma~\ref{lem:trans-inv} by using the \textsc{trans-inv} rule on $C_1 \parallel C_2$.
\end{proof}
\begin{theorem}
    \label{thm:bridge-comp}
    \textsc{bridge-comp} is sound.
\end{theorem}
\begin{proof}
    Suppose that $I_1 \vdash \ag{A} C_1 \parallel B \ag{R}$ and $I_2 \vdash \ag{R} B \parallel C_2 \ag{G}$.
    By Lemma~\ref{lem:interf-free} and the fact that $\parallel$ is commutative, we have $I_1 \vdash \ag{A} C_1 \parallel B \parallel C_2 \ag{R}$ and $I_2 \vdash \ag{R} C_1 \parallel B \parallel C_2 \ag{G}$.
    Finally, the theorem follows with Lemma~\ref{lem:trans-inv} by using the \textsc{trans-inv} rule on $C_1 \parallel B \parallel C_2$.
\end{proof}
\begin{theorem}
    \textsc{aux-comp} is sound.
\end{theorem}
\begin{proof}
    Suppose $I_1 \vdash \ag{A} C_1 \parallel B \ag{R}$, $I_2 \vdash \ag{R} B \parallel C_2 \ag{G}$, and $Aux\ B$.
    By Thm~\ref{thm:bridge-comp} and Thm.~\ref{thm:ii-implies-safe}, we have $\vdash \ag{A} C_1 \parallel B \parallel C_2 \ag{G}$.
    Finally, because $Aux\ B$, we have $\vdash \ag{A} C_1 \parallel C_2 \ag{G}$ by Def.~\ref{def:aux-component}.
\end{proof}

\section{Proof of Thm.~\ref{thm:b-is-aux}}
\label{apx:b-is-aux}
We now prove Thm.~\ref{thm:b-is-aux} which states: Let $\Box\phi$ be an action invariant, then $\mathcal{B}(\Box\phi)$ is an auxiliary component.
\begin{proof}
    Let $\ag{A} C \ag{G}$ be a state-based contract such that $\vdash \ag{A} C \parallel \mathcal{B}(\Box\phi) \ag{G}$; by Def.~\ref{def:aux-component}, the proof obligation is to show that $\vdash \ag{A} C \ag{G}$.
    Let $C = (vars,Init,Next)^p$.
    By Def.~\ref{def:ag-contract} and the definition of parallel composition, we can infer that $(vars,Init \land A, Next \land A')^p \parallel \mathcal{B}(\Box\phi) \models \Box G$.
    However, all actions in $\mathcal{B}(\Box\phi)$ are enabled in every state, meaning that $\mathcal{B}(\Box\phi)$ allows all possible behaviors.
    We can therefore infer $(vars,Init \land A, Next \land A')^p \models \Box G$, from which the theorem follows by Def.~\ref{def:ag-contract}.
\end{proof}

\section{Proof of Thm.~\ref{thm:sfl-aut-equiv}}
\label{apx:proof-sfl-aut-equiv}
We now dedicate this appendix to proving that action invariants are semantically equivalent to their transition system counterpart given by $\mathcal{T}$.
Formally, the proof obligation is to show that $\mathcal{L}(\Box\phi) = \mathcal{L}(\mathcal{T}(\Box\phi))$, where $\Box\phi$ is an arbitrary action invariant.
This is exactly what we will prove in Thm.~\ref{thm:sfl-aut-equiv} at the end of this subsection; however, we first prove two helper lemmas.
\begin{lemma}
    \label{lem:exists-st-tr}
    There exists a unique state-based behavior in $\mathcal{B}(\Box\phi)$ that corresponds to executing the actions in $\sigma$.
\end{lemma}
\begin{proof}
    The behavior exists because all actions are always enabled in $\mathcal{B}(\Box\phi)$ and the behavior is unique because $\mathcal{B}(\Box\phi)$ is deterministic.
\end{proof}
\begin{lemma}
    \label{lem:act-to-st-tr}
    Let $\tau$ be the state-based behavior in $\mathcal{B}(\Box\phi)$ that corresponds to executing the actions in $\sigma$ (by Lemma~\ref{lem:exists-st-tr}).
    Then, $\sigma \models \Box\phi$ if and only if $\tau \models \Box\mathcal{R}(\phi)$.
\end{lemma}
\begin{proof}
    The proof obligation is to show that for all $i \geq 0$, we have $\emptyset \vdash \sigma,i \models \phi$ if and only if $\emptyset \vdash \tau_i \models \mathcal{R}(\phi)$.
    However, it suffices to prove the following stronger statement: for all well-formed environments $E$ (environments that assign values to the free variables in $\phi$) and for all $i \geq 0$, we have $E \vdash \sigma,i \models \phi$ if and only if $E \vdash \tau_i \models \mathcal{R}(\phi)$.
    We will prove this statement by structural induction for \flogic{} over $\phi$.
    
    In the base case, we have $\phi = f(r)$, where $f = (s,v,T)$ is a fluent and $r$ are arguments to the fluent.
    Let a well-formed environment $E$ and $i \geq 0$ be given.
    Furthermore, let $v_i$ be the value of the variable $v$ in the state $\tau_i$.
    Then, $E \vdash \sigma,i \models f(r)$ if and only if $(f|\sigma_0\dots\sigma_i)[r[E]]$ if and only if $v_i[r[E]]$ if and only if $E \vdash \tau_i \models \mathcal{R}(f(r))$.
    
    Since $\phi$ is a non-temporal formula, it suffices to only consider non-temporal connectives for the inductive step.
    In the case that $\phi = \psi_1 \lor \psi_2$, let a well-formed environment $E$ (for $\phi$) and $i \geq 0$ be given.
    $E$ must also be well-formed for $\psi_1$ and $\psi_2$, since no free-variables are introduced by the disjunction.
    Therefore, we can invoke the inductive hypothesis to show that $E \vdash \sigma,i \models \psi_1$ if and only if $E \vdash \tau_i \models \mathcal{R}(\psi_1)$ and also $E \vdash \sigma,i \models \psi_2$ if and only if $E \vdash \tau_i \models \mathcal{R}(\psi_2)$.
    However, this allows us to conclude that $E \vdash \sigma,i \models \phi$ if and only if $E \vdash \tau_i \models \mathcal{R}(\phi)$ by the definition of $\phi$ and also the definition of $\mathcal{R}$ over disjunctions.
    The remaining connectives are similar.
\end{proof}
We now provide a proof for Thm.~\ref{thm:sfl-aut-equiv}, which states that for any action invariant $\Box\phi$, we have $\mathcal{L}(\Box\phi) = \mathcal{L}(\mathcal{T}(\Box\phi))$.
\begin{proof}
    We will prove that for any action-based behavior $\sigma$, $\sigma \models \Box\phi$ if and only if $\sigma \models \mathcal{T}(\Box\phi)$.
    Let $\tau$ be the state-based behavior in $\mathcal{B}(\Box\phi)$ that corresponds to executing the actions in $\sigma$ (by Lemma~\ref{lem:exists-st-tr}).
    \begin{align*}
        \sigma \models \Box\phi &\text{ if and only if } \tau \models \Box\mathcal{R}(\phi) && \text{By Lemma~\ref{lem:act-to-st-tr}.} \\
        \tau \models \Box\mathcal{R}(\phi) &\text{ if and only if } \tau \models \mathcal{T}(\Box\phi) && \text{By Def.~\ref{def:sfl-st}.} \\
        \tau \models \mathcal{T}(\Box\phi) &\text{ if and only if } \sigma \models \mathcal{T}(\Box\phi) && \text{By the assumption that $\tau$ corresponds to $\sigma$ in} \\
        & && \text{$\mathcal{B}(\Box\phi)$, and hence also $\mathcal{T}(\Box\phi)$. For ``only if'',}\\
        & && \text{also because $\mathcal{T}(\Box\phi)$ is deterministic.}
    \end{align*}
    Together, the equations above imply the intended result.
\end{proof}

\section{Proof of Thm.~\ref{thm:action-ii-sound}}
\label{apx:action-ii-sound}
We now prove Thm.~\ref{thm:action-ii-sound} which states: $I \vdash \sag{\alpha} C \sag{\gamma}$ implies $\vdash \sag{\alpha} C \sag{\gamma}$
\begin{proof}
    Assume that $I \vdash \sag{\alpha} C \sag{\gamma}$.
    Then, from both Def.~\ref{def:sfl-local-ii} and Thm.~\ref{thm:ii-implies-safe}, we can infer that $\vdash \ag{\mathcal{R}(\alpha)} \mathcal{B}(\alpha) \parallel C \parallel \mathcal{B}(\gamma) \ag{\mathcal{R}(\gamma)}$.
    Let $\mathcal{B}(\alpha) = (vars,Init,Next)^p$.
    By Def.~\ref{def:ag-contract} and the definition of parallel composition we can infer $(vars, Init \land \mathcal{R}(\alpha), Next \land \mathcal{R}(\alpha)')^p \parallel C \parallel \mathcal{B}(\gamma) \models \Box \mathcal{R}(\gamma)$.
    By Def.~\ref{def:sfl-st}, we see that this is equivalent to the statement $\mathcal{T}(\Box\alpha) \parallel C \parallel \mathcal{B}(\gamma) \models \Box \mathcal{R}(\gamma)$.

    Now consider an action-based behavior $\sigma$ such that $\sigma \models \mathcal{T}(\Box\alpha) \parallel C$; we will show that $\sigma \models \mathcal{T}(\Box\gamma)$ to complete the proof.
    Because $\mathcal{B}(\gamma)$ has every action enabled in all states, it must also be the case that $\sigma \models \mathcal{T}(\Box\alpha) \parallel C \parallel \mathcal{B}(\gamma)$.
    By Lemma~\ref{lem:exists-st-tr}, there exists a unique state based behavior $\tau$ in $\mathcal{B}(\gamma)$ that corresponds to $\sigma$.
    This in turn implies that $\tau \models \mathcal{T}(\Box\alpha) \parallel C \parallel \mathcal{B}(\gamma)$, which also implies that $\tau \models \Box \mathcal{R}(\gamma)$.
    Because $\tau \models \mathcal{B}(\gamma)$ and $\tau \models \Box\mathcal{R}(\gamma)$, we can infer that $\tau \models \mathcal{T}(\Box\gamma)$.
    Finally, by the assumption that $\sigma$ corresponds to $\tau$, we see that $\sigma \models \mathcal{T}(\Box\gamma)$.
\end{proof}

\section{Soundness of Action-Based Contract Composition}
\label{apx:action-composition-sound}
\begin{theorem}
    \label{thm:sfl-comp-sound}
    Rule \textsc{sfl-comp} is sound.
\end{theorem}
\begin{proof}
    Suppose that $I_1 \vdash \sag{\alpha} C_1 \sag{\rho}$ and $I_2 \vdash \sag{\rho} C_2 \sag{\gamma}$.
    By Def.~\ref{def:sfl-ag-contract}, we also have $I_1 \vdash \ag{\mathcal{R}(\alpha)} \mathcal{B}(\alpha) \parallel C_1 \parallel \mathcal{B}(\rho) \ag{\mathcal{R}(\rho)}$ and $I_2 \vdash \ag{\mathcal{R}(\rho)} \mathcal{B}(\rho) \parallel C_2 \parallel \mathcal{B}(\gamma) \ag{\mathcal{R}(\gamma)}$.
    Using the \textsc{bridge-comp} inference rule (Thm.~\ref{thm:bridge-comp}), we can infer that $I_1 \land I_2 \vdash \ag{\mathcal{R}(\alpha)} \mathcal{B}(\alpha) \parallel C_1 \parallel \mathcal{B}(\rho) \parallel C_2 \parallel \mathcal{B}(\gamma) \ag{\mathcal{R}(\gamma)}$.
    Finally, we see that $I_1 \land I_2 \vdash \sag{\alpha} C_1 \parallel \mathcal{B}(\rho) \parallel C_2 \sag{\gamma}$ by Def.~\ref{def:sfl-ag-contract}.
\end{proof}
\begin{theorem}
    Rule \textsc{sfl-safe} is sound.
\end{theorem}
\begin{proof}
    Suppose that $I_1 \vdash \sag{\alpha} C_1 \sag{\rho}$ and $I_2 \vdash \sag{\rho} C_2 \sag{\gamma}$.
    By Thm.~\ref{thm:sfl-comp-sound}, we can infer $I_1 \land I_2 \vdash \sag{\alpha} C_1 \parallel \mathcal{B}(\rho) \parallel C_2 \sag{\gamma}$.
    By Def.~\ref{def:sfl-ag-contract}, we also have $I_1 \land I_2 \vdash \ag{\mathcal{R}(\alpha)} \mathcal{B}(\alpha) \parallel C_1 \parallel \mathcal{B}(\rho) \parallel C_2 \parallel \mathcal{B}(\gamma) \ag{\mathcal{R}(\gamma)}$.
    However, by Thm.~\ref{thm:ii-implies-safe}, we can also infer $\vdash \ag{\mathcal{R}(\alpha)} \mathcal{B}(\alpha) \parallel C_1 \parallel \mathcal{B}(\rho) \parallel C_2 \parallel \mathcal{B}(\gamma) \ag{\mathcal{R}(\gamma)}$.
    The theorem then follows by Def.~\ref{def:sfl-ag-contract}.
\end{proof}

\section{Proof of Thm.~\ref{thm:hybrid-ii-sound}}
\label{apx:hybrid-ii-sound}
We now prove Thm.~\ref{thm:hybrid-ii-sound} which states: $I \vdash \sag{\alpha} C \sag{G}$ implies $\vdash \sag{\alpha} C \sag{G}$
\begin{proof}
    Assume that $I \vdash \sag{\Box\alpha} C \sag{G}$.
    Then, from both Def.~\ref{def:sfl-local-ii} and Thm.~\ref{thm:ii-implies-safe}, we can infer that $\vdash \ag{\mathcal{R}(\alpha)} \mathcal{B}(\Box\alpha) \parallel C \ag{G}$.
    By Def.~\ref{def:ag-contract} and the definition of parallel composition, we can also infer $(vars, Init \land \mathcal{R}(\alpha), Next \land \mathcal{R}(\alpha)')^p \parallel C \models \Box G$, where $\mathcal{B}(\Box\alpha) = (vars,Init,Next)^p$.
    By Def.~\ref{def:sfl-st}, we see that this is equivalent to the statement $\mathcal{T}(\Box\alpha) \parallel C \models \Box G$.
    Finally, the theorem is proved by Def.~\ref{def:sfl-ag-contract}.
\end{proof}

\section{Soundness of Hybrid Contract Composition}
\label{apx:hybrid-composition-sound}
We now prove that the two inference rules \textsc{hybrid-comp} and \textsc{hybrid-safe} are sound.
\begin{theorem}
    \label{thm:hybrid-comp-sound}
    The rule \textsc{hybrid-comp} is sound.
\end{theorem}
\begin{proof}
    Suppose that $I_1 \vdash \sag{\alpha} C_1 \sag{\rho}$ and $I_2 \vdash \sag{\rho} C_2 \sag{G}$.
    By Def.~\ref{def:sfl-ag-contract}, we have $I_1 \vdash \ag{\mathcal{R}(\alpha)} \mathcal{B}(\alpha) \parallel C_1 \parallel \mathcal{B}(\rho) \ag{\mathcal{R}(\rho)}$ and by Def.~\ref{def:hybrid-ag-contract} we also have $I_2 \vdash \ag{\mathcal{R}(\rho)} \mathcal{B}(\rho) \parallel C_2 \ag{G}$.
    Using the \textsc{bridge-comp} inference rule (Thm.~\ref{thm:bridge-comp}), we can infer that $I_1 \land I_2 \vdash \ag{\mathcal{R}(\alpha)} \mathcal{B}(\alpha) \parallel C_1 \parallel \mathcal{B}(\rho) \parallel C_2 \ag{G}$.
    Finally, we see that $I_1 \land I_2 \vdash \sag{\alpha} C_1 \parallel \mathcal{B}(\rho) \parallel C_2 \sag{G}$ by Def.~\ref{def:hybrid-ag-contract}.
\end{proof}
\begin{theorem}
    The rule \textsc{hybrid-safe} is sound.
\end{theorem}
\begin{proof}
    Suppose that $I_1 \vdash \sag{\alpha} C_1 \sag{\rho}$ and $I_2 \vdash \sag{\rho} C_2 \sag{G}$.
    By Thm.~\ref{thm:hybrid-comp-sound}, we can infer $I_1 \land I_2 \vdash \sag{\alpha} C_1 \parallel \mathcal{B}(\rho) \parallel C_2 \sag{G}$.
    By Def.~\ref{def:hybrid-ag-contract}, we also have $I_1 \land I_2 \vdash \ag{\mathcal{R}(\alpha)} \mathcal{B}(\alpha) \parallel C_1 \parallel \mathcal{B}(\rho) \parallel C_2 \ag{G}$.
    However, by Thm.~\ref{thm:ii-implies-safe}, we can also infer $\vdash \ag{\mathcal{R}(\alpha)} \mathcal{B}(\alpha) \parallel C_1 \parallel \mathcal{B}(\rho) \parallel C_2 \ag{G}$.
    The theorem then follows by Def.~\ref{def:hybrid-ag-contract}.
\end{proof}

\section{Bridge and Fluent Definitions for Mongo Static Raft}
\label{apx:mongo-fluent-defs}
\begin{figure}
    \begin{subfigure}[b]{0.48\textwidth}
        \vspace{2mm}
        {\small \input{tla/msr_committedThisTerm}}
        \caption{Fluent definition for \textit{committedThisTerm}.}
        \label{fig:mongo-committedThisTerm}
    \end{subfigure}
    \hfill
    \begin{subfigure}[b]{0.47\textwidth}
        \vspace{2mm}
        {\small \input{tla/msr_globalCurrentTerm}}
        \caption{Fluent definition for \textit{globalCurrentTerm}.}
        \label{fig:mongo-globalCurrentTerm}
    \end{subfigure}
    \begin{subfigure}[b]{0.48\textwidth}
        \vspace{2mm}
        {\small \input{tla/msr_reqThisTerm}}
        \caption{Fluent definition for \textit{reqThisTerm}.}
        \label{fig:mongo-reqThisTerm}
    \end{subfigure}
    \begin{subfigure}[t]{0.99\textwidth}
        \begin{align*}
            &\land\ \forall t \in \mathbb{N} : \forall s \in Server : committedThisTerm(t,s) \implies globalCurrentTerm(t) \\
            &\land\ \forall t \in \mathbb{N} : \forall s \in Server : reqThisTerm(s,t) \implies currentLeader(s,t) \\
            &\land\ \forall t \in \mathbb{N} : \forall s \in Server : committedThisTerm(t,s) \implies currentLeader(s,t) \\
            &\land\ \forall t \in \mathbb{N} : \A s \in Server : currentLeader(s,t) \implies leaderAtTermServ(t,s) \\
            &\land\ \forall t_1,t_2 \in \mathbb{N} : \forall s \in Server : \\
            &\hspace{3mm}(leaderAtTermServ(t_1,s) \land globalCurrentTerm(t_2)) \implies (t_1 \leq t_2) \\
            &\land\ \forall t \in \mathbb{N} : \forall s_1,s_2 \in Server : \\
            &\hspace{3mm}(leaderAtTermServ(t,s_1) \land leaderAtTermServ(t,s_2)) \implies (s_1 = s_2)
        \end{align*}
        \caption{The entire bridge formula $\rho_l$ for the $Log$ and $StateTerm$ components, including the \textit{one leader per term} conjunct.}
        \label{fig:mongo-log-bridges}
    \end{subfigure}
    \caption{Additional definitions for the MongoStaticRaft protocol.}
    \label{fig:mongo-apx}
\end{figure}

\end{document}